\documentclass[12pt]{amsart}
\usepackage{}

\usepackage{amsmath}
\usepackage{amsfonts}
\usepackage{amssymb}
\usepackage[all]{xy}           

\usepackage{bbding}
\usepackage{txfonts}
\usepackage{amscd}

\usepackage[shortlabels]{enumitem}
\usepackage{ifpdf}
\ifpdf
  \usepackage[colorlinks,final,backref=page,hyperindex]{hyperref}
\else
  \usepackage[colorlinks,final,backref=page,hyperindex,hypertex]{hyperref}
\fi
\usepackage{tikz}
\usepackage[active]{srcltx}

\makeatletter

\newtheorem{thm}{Theorem}[section]
\newtheorem{lem}[thm]{Lemma}
\newtheorem{cor}[thm]{Corollary}
\newtheorem{pro}[thm]{Proposition}
\newtheorem{ex}[thm]{Example}
\newtheorem{rmk}[thm]{Remark}
\newtheorem{defi}[thm]{Definition}

\newcommand {\emptycomment}[1]{}

\newcommand{\B}{\mathsf{B}}

\newcommand{\lon }{\,\rightarrow\,}
\newcommand{\be }{\begin{equation}}
\newcommand{\ee }{\end{equation}}

\newcommand{\g}{\mathfrak g}

\newcommand{\huaB}{\mathcal{B}}

\newcommand{\huaL}{\mathcal{L}}
\newcommand{\huaR}{\mathcal{R}}

\newcommand{\huaG}{\mathcal{G}}

\newcommand{\huaO}{{\mathcal{O}}}

\newcommand{\frkk}{\mathfrak k}

\newcommand{\frks}{\mathfrak s}

\newcommand{\half}{\frac{1}{2}}

\newcommand{\Id}{\rm{Id}}

\newcommand{\br}[1]{   [ \cdot,    \cdot  ]   }

\newcommand{\Hom}{\mathrm{Hom}}

\newcommand{\Sym}{\mathrm{Sym}}

\newcommand{\gl}{\mathfrak {gl}}

\newcommand{\K}{\mathbb{K}}

\begin{document}

\title[Leibniz bialgebras  ]{Leibniz bialgebras, relative Rota-Baxter operators and the classical Leibniz  Yang-Baxter equation}

\author{Yunhe Sheng}
\address{Department of Mathematics, Jilin University, Changchun 130012, Jilin, China}
\email{shengyh@jlu.edu.cn}

\author{Rong Tang}
\address{Department of Mathematics, Jilin University, Changchun 130012, Jilin, China}
\email{tangrong@jlu.edu.cn}


\begin{abstract}

In this paper, first we introduce the notion of a Leibniz bialgebra and show that matched pairs of Leibniz algebras, Manin triples of Leibniz algebras and Leibniz bialgebras are equivalent. Then we introduce the notion of a (relative) Rota-Baxter operator on a Leibniz algebra and construct the graded Lie algebra that characterizes relative Rota-Baxter operators as Maurer-Cartan elements. By these structures and the  twisting theory of twilled Leibniz algebras, we
 further define the classical Leibniz Yang-Baxter equation, classical Leibniz $r$-matrices and  triangular Leibniz bialgebras. Finally, we construct solutions of the classical Leibniz Yang-Baxter equation using relative Rota-Baxter operators and  Leibniz-dendriform algebras.

\end{abstract}



\maketitle

\tableofcontents

\allowdisplaybreaks


\section{Introduction}
\label{sec:intr}

The paper aims to establish the bialgebra theory for Leibniz algebras. In particular, define what is a triangular Leibniz bialgebra, what is a classical Leibniz Yang-Baxter equation  and what is a classical Leibniz $r$-matrix.

\subsection{Leibniz algebras and Leibniz bialgebras}
The notion of a Leibniz algebra was introduced by Loday \cite{Loday,Loday and Pirashvili} with the motivation in the study of the periodicity in algebraic K-theory.   Recently Leibniz algebras were studied from different aspects due to applications in both mathematics and physics.
In particular, integration of Leibniz algebras were studied in \cite{BW,Int1} and deformation quantization of Leibniz algebras was studied in \cite{DW}. As the underlying structure of embedding tensor, Leibniz algebras also have application in higher gauge theories, see \cite{KS,SW} for more details.

For a given algebraic structure, a bialgebra structure on this algebra is obtained
by a  comultiplication together with some compatibility conditions between the multiplication and the comultiplication.
 A good compatibility condition is prescribed   by a rich structure theory and effective constructions.
The most famous examples of bialgebras are  associative bialgebras and Lie bialgebras, which have important applications in both mathematics and mathematical physics, e.g. a Lie bialgebra is the
algebraic structure corresponding to a Poisson-Lie group and the classical structure of a quantized
universal enveloping algebra~\cite{CP,D}.

The purpose of this paper is to study the bialgebra theory for Leibniz algebras with the motivation from the great importance of Lie bialgebras. It is well known that a Lie bialgebra is equivalent to a Manin tripe of Lie algebras. In the definition of a Manin triple, one needs to use a quadratic Lie algebra, which is a Lie algebra equipped with a symmetric nondegenerate invariant bilinear form. However, to define a quadratic Leibniz algebra, we need to use a skew-symmetric bilinear form. This is supported by the fact that the operad of Lie algebras is a cyclic operad, but the operad of Leibniz algebras is an anticyclic operad.  Actually, it is observed by Chapoton in \cite{Chapoton} using the operad theory that one should use the aforementioned skew-symmetric invariant  bilinear form on a Leibniz algebra. As soon as we have the correct notion of a quadratic Leibniz algebra, we can define Manin triples and dual representations of Leibniz algebras. We introduce the notion of a Leibniz bialgebra and show that matched pairs of Leibniz algebras, Manin triples of Leibniz algebras and Leibniz bialgebras are equivalent. Even though we obtain some nice results totally parallel to the context of Lie bialgebras, we need to emphasize that our bialgebra theory are not generalization of Lie bialgebras, namely the restriction of our theory on Lie algebras is independent of Lie bialgebras.

\subsection{Triangular Leibniz bialgebras: relative Rota-Baxter operator approach}
Due to the importance of the classical Yang-Baxter equation and triangular Lie bialgebras, it is natural to define the Leibniz analogue of the classical Yang-Baxter equation and triangular Leibniz bialgebras.
This is a very hard problem due to that the representation theory of Leibniz algebras is not good, e.g. there is no tensor product in the module category of Leibniz algebras. We solve this problem using relative Rota-Baxter operators and the twisting theory of twilled Leibniz algebras.

  A Rota-Baxter
operator  on a Lie algebra
was introduced in the
1980s as the operator form of the classical Yang-Baxter equation.
To better understand the classical Yang-Baxter equation and
the related integrable systems, the more general notion of an $\huaO$-operator
on a Lie algebra was introduced by Kupershmidt~\cite{Ku},
which can be traced back to Bordemann
\cite{Bor}.  An
$\huaO$-operator gives rise to a skew-symmetric $r$-matrix in a larger Lie algebra \cite{Bai-1}.
In the context of associative algebras, $\huaO$-operators give rise to
  dendriform algebras \cite{Lo5}, play important role in the bialgebra theory \cite{Bai-2} and  lead to the splitting of
operads~\cite{Bai-Bellier-Guo-Ni}. See the book ~\cite{Gub} for more details.

The twisting theory was introduced by Drinfeld in \cite{Dri} motivated by the study of quasi-Lie bialgebras and quasi-Hopf algebras. As a useful tool in the study of bialgebras, the twisting theory was further applied to associative algebras and Poisson geometry, see \cite{Schwarzbach1,Schwarzbach2,Roy,Uchino} for more details.

 In this paper, we introduce the notion of a relative Rota-Baxter operator  on a Leibniz algebra. 
  We construct the graded Lie algebra that characterize relative Rota-Baxter operators as Maurer-Cartan elements. Using this graded Lie algebra, we give the definition of a classical  Leibniz Yang-Baxter equation. Moreover, we give the twisting theory of twilled Leibniz algebras, by which we define triangular Leibniz bialgebras.  We also use relative Rota-Baxter operators and Leibniz-dendriform algebras to give solutions of the classical Leibniz Yang-Baxter equations in some larger Leibniz algebras. \vspace{-5mm}

\subsection{Outline of the paper}

In Section \ref{sec:B}, we introduce the notions of a Manin triple of Leibniz algebras and a Leibniz bialgebra.   We prove the equivalence between matched pairs of Leibniz algebras, Manin triples of Leibniz algebras and   Leibniz bialgebras. The main innovation is that we use skew-symmetric invariant bilinear form instead of symmetric invariant bilinear form  in the definition of a quadratic Leibniz algebra. Another ingredient is that the dual representation of a representation $(V;\rho^L,\rho^R)$ should be $(V^*;(\rho^L)^*,-(\rho^L)^*-(\rho^R)^*)$, rather than $(V^*;(\rho^L)^*,(\rho^R)^*)$ under some conditions.
Therefore, our approach is totally different from the one in \cite{Barreiro-Benayadi}, where the authors have to use symmetric Leibniz algebras.  Thus,   Leibniz algebras should be considered as an independent algebraic structure from Lie algebras, not viewed as a simple noncommutative generalization of Lie algebras.

In Section \ref{sec:K}, we make preparations for our later study of triangular Leibniz algebras. In Section \ref{sec:K1}, we give the graded Lie algebra that characterize Leibniz algebras as Maurer-Cartan elements and some technical tools. In Section \ref{sec:K2}, we introduce the notion of a relative Rota-Baxter operator on a Leibniz algebra with respect to a representation and construct the graded Lie algebra that characterize it as a Maurer-Cartan element. This is the foundation of the whole paper. In Section \ref{sec:K3}, we introduce the notions of a twilled Leibniz algebra.  The twisting theory of twilled Leibniz algebras is studied in detail for the purpose to define triangular Leibniz bialgebras.

In Section \ref{sec:R}, we study triangular Leibniz bialgebras. The traditional coboundary approach for Lie bialgebras does not work for Leibniz bialgebras because there is no tensor product for two representations of a Leibniz algebra. Thus, one need to use new ideas and new methods to solve this problem.  We define the classical Leibniz Yang-Baxter equation and  a classical Leibniz $r$-matrix using the graded Lie algebra given in   Section \ref{sec:K2},  and then define a triangular Leibniz bialgebra  successfully using the twisting theory of a twilled Leibniz algebra given in   Section \ref{sec:K3}.   We also   generalize a Semonov-Tian-Shansky's result about the relation between the operator form and the tensor form  of a classical $r$-matrix in \cite{STS} to the context of Leibniz algebras.

In Section \ref{sec:S}, first we use relative Rota-Baxter operators to give solutions of the classical Leibniz Yang-Baxter equation in the semidirect product Leibniz algebra. Then we introduce the notion of a Leibniz-dendriform algebra. Similar to the connection from pre-Lie algebras to Lie algebras and from dendriform algebras to associative algebras, we show that a Leibniz-dendriform algebra gives rise to  a Leibniz algebra together with a representation on itself.  The importance of such a structure is due to that it gives rise to a solution of the classical Leibniz Yang-Baxter equation.

 Quantization of Lie bialgebras and deformation  quantization of Leibniz algebras  was studied in \cite{DW,EK}. It is natural to study  quasi-triangular Leibniz bialgebras and their quantization. On the other hand, classical $r$-matrices play important role in the study of integrable systems. It is natural to investigate whether classical Leibniz $r$-matrices can be applied to some integrable systems. It is also natural to investigate the global objects corresponding to Leibniz bialgebras.  We will study these questions in the future.

\vspace{2mm}

In this paper, we work over an algebraically closed field $\K$ of characteristic 0 and all the vector spaces are over $\K$ and finite-dimensional.

\section{Quadratic Leibniz algebras and Leibniz bialgebras}\label{sec:B}

\begin{defi}
A {\bf Leibniz algebra} is a vector space $\g$ together with a bilinear operation $[\cdot,\cdot]_\g:\g\otimes\g\lon\g$ such that
\begin{eqnarray*}
\label{Leibniz}[x,[y,z]_\g]_\g=[[x,y]_\g,z]_\g+[y,[x,z]_\g]_\g,\quad\forall x,y,z\in\g.
\end{eqnarray*}
\end{defi}

A {\bf representation} of a Leibniz algebra $(\g,[\cdot,\cdot]_{\g})$ is a triple $(V;\rho^L,\rho^R)$, where $V$ is a vector space, $\rho^L,\rho^R:\g\lon\gl(V)$ are linear maps such that the following equalities hold for all $x,y\in\g$,
\begin{eqnarray}
\label{rep-1}\rho^L([x,y]_{\g})&=&[\rho^L(x),\rho^L(y)],\\
\label{rep-2}\rho^R([x,y]_{\g})&=&[\rho^L(x),\rho^R(y)],\\
\label{rep-3}\rho^R(y)\circ \rho^L(x)&=&-\rho^R(y)\circ \rho^R(x).
\end{eqnarray}
Here $[\cdot,\cdot]:\wedge^2\gl(V)\lon\gl(V)$ is the commutator Lie bracket on $\gl(V)$, the vector space of linear transformations on $V$.

Define the left multiplication $L:\g\longrightarrow\gl(\g)$ and the right multiplication $R:\g\longrightarrow\gl(\g)$ by $L_xy=[x,y]_\g$ and $R_xy=[y,x]_\g$ respectively for all $x,y\in \g$.  Then $(\g;L,R)$ is a representation of $(\g,[\cdot,\cdot]_{\g})$, which is called the {\bf regular representation}. Define two linear maps $L^*,R^*:\g\longrightarrow\gl(\g^*)$ with $x\longrightarrow L^*_x$ and $x\longrightarrow R^*_x$ respectively by
\begin{eqnarray}
\langle L^*_{x}\xi,y\rangle=-\langle \xi,[x,y]_\g\rangle,\quad\langle R^*_{x}\xi,y\rangle=-\langle \xi,[y,x]_\g\rangle,\quad\forall x,y\in\g,\xi\in\g^*.
\end{eqnarray}
If there is a Leibniz algebra structure on the dual space $\g^*$, we denote the left multiplication and the right multiplication by $\huaL$ and $\huaR$ respectively.

\subsection{Quadratic Leibniz algebras and the Leibniz analogue of the string Lie 2-algebra}
It is observed by Chapoton in \cite{Chapoton} using the operad theory  that one need to use skew-symmetric bilinear forms instead of symmetric bilinear forms on a Leibniz algebra. This is the key ingredient in our study of Leibniz bialgebras.

\begin{defi}{\rm (\cite{Chapoton})}
A {\bf quadratic Leibniz algebra} is a Leibniz algebra $(\g,[\cdot,\cdot]_\g)$ equipped with a nondegenerate skew-symmetric bilinear form $\omega\in\wedge^2\g^*$ such that the following invariant condition
holds:
\begin{eqnarray}\label{Invariant-bilinear-forms}
\omega(x,[y,z]_\g)=\omega([x,z]_\g+[z,x]_\g,y),\quad \forall  x,y,z\in \g.
\end{eqnarray}
\end{defi}

\begin{rmk}
In the original definition of a nondegenerate skew-symmetric invariant bilinear form on a Leibniz algebra $(\g,[\cdot,\cdot]_\g)$ given  in \cite{Chapoton}, there is a superfluous condition
\begin{eqnarray}\label{superfluous-condition}
\omega(x,[y,z]_\g)=-\omega([y,x]_\g,z).
\end{eqnarray}
In fact, by \eqref{Invariant-bilinear-forms},   we have
\begin{eqnarray*}
-\omega([y,x]_\g,z)=\omega(z,[y,x]_\g)=\omega([z,x]_\g+[x,z]_\g,y)=\omega(x,[y,z]_\g),\quad \forall x,y,z\in\g.
\end{eqnarray*}
\end{rmk}

\begin{rmk}
  Note that we use skew-symmetric bilinear forms instead of symmetric bilinear forms and use the invariant condition \eqref{Invariant-bilinear-forms} instead of the invariant condition $ B([x,y]_\g,z)=B(x,[y,z]_\g)$ and this is the main ingredient in our study of Leibniz bialgebras. In \cite{Barreiro-Benayadi}, the author  use symmetric bilinear form and invariant condition $ B([x,y]_\g,z)=B(x,[y,z]_\g)$ to study Leibniz bialgebras so that one has to add some strong conditions. As we will see, everything in the following study is natural in the sense that we do not need to add any extra conditions on the Leibniz algebra.
\end{rmk}

Recall that a quadratic Lie algebra is a Lie algebra $(\frkk,[\cdot,\cdot]_\frkk)$ equipped with a nondegenerate symmetric bilinear form $B\in\Sym^2(\frkk^*)$, which is invariant in the sense that
$$
B([x,y]_\g,z)=B(x,[y,z]_\g),\quad \forall x,y,z\in\frkk.
$$
Associated to a quadratic Lie algebra $(\frkk,[\cdot,\cdot]_\frkk,B)$, we have a closed 3-form $\bar{\Theta}\in\wedge^3\frkk^*$ given by
$$
\bar{\Theta}(x,y,z)=B(x,[y,z]_\frkk),
$$
which is known as the Cartan 3-form.

Let $(\g,[\cdot,\cdot]_\g,\omega)$ be a quadratic Leibniz algebra. Define $\Theta\in\otimes^3\g^*$ by
\begin{equation}
  \Theta(x,y,z)=\omega(x,[y,z]_\g),\quad\forall x,y,z\in\g.
\end{equation}
This 3-tensor can be viewed as the Leibniz analogue of the Cartan 3-form on a quadratic Lie algebra as the following lemma shows.

\begin{lem}\label{lem:3-form}
  With the above notations, $\Theta$ is a $3$-cocycle on the Leibniz algebra $(\g,[\cdot,\cdot]_\g)$ with values in the trivial representation $(\K;0,0)$, i.e. $\partial \Theta=0.$
\end{lem}
\begin{proof}
  For all $x,y,z,w\in\g$, by the fact that $[\g,\g]_\g$ is a left center, we have
  \begin{eqnarray*}
    (\partial \Theta)(x,y,z,w)&=&-\Theta([x,y]_\g,z,w)-\Theta(y,[x,z]_\g,w)-\Theta(y,z,[x,w]_\g)\\
    &&+\Theta(x,[y,z]_\g,w)+\Theta(x,z,[y,w]_\g)-\Theta(x,y,[z,w]_\g)\\
    &=&-\omega([x,y]_\g,[z,w]_\g)-\omega(y,[[x,z]_\g,w]_\g)-\omega(y,[z,[x,w]_\g]_\g)\\
    &&+\omega(x,[[y,z]_\g,w]_\g)+\omega(x,[z,[y,w]_\g]_\g)-\omega(x,[y,[z,w]_\g]_\g)\\
    &=&\omega(y,[x,[z,w]_\g]_\g)-\omega(y,[[x,z]_\g,w]_\g)-\omega(y,[z,[x,w]_\g]_\g)\\
    &&+\omega(x,[[y,z]_\g+[z,y]_\g,w]_\g)\\
    &=&0,
  \end{eqnarray*}
  which finishes the proof.
\end{proof}

Consequently, given a quadratic Leibniz algebra $(\g,[\cdot,\cdot]_\g,\omega)$, we can construct a Leibniz 2-algebra (2-term $Lod_\infty$-algebra) \cite{ammardefiLeibnizalgebra,livernet,LiuSheng}, which can be viewed as the Leibniz analogue of the string Lie 2-algebra associated to a semisimple Lie algebra \cite{baez:2algebras}.

On the graded vector space $\K\oplus\g$, define $l_1:\K\lon \g$ to be the zero map, define $l_2$ and $l_3$ by
$$
\left\{\begin{array}{rcll}
l_2(x,y)&=&[x,y]_\g, &\forall x,y\in\g,\\
l_2(x,s)&=&l_2(s,x)=0,&\forall x\in\g,~s\in\K,\\
l_3(x,y,z)&=&\Theta(x,y,z)=\omega(x,[y,z]_\g),&\forall x,y,z\in\g.
\end{array}\right.
$$
\begin{thm}
Let $(\g,[\cdot,\cdot]_\g,\omega)$ be a quadratic Leibniz algebra. Then $(\K,\g,l_1=0,l_2,l_3)$ is a Leibniz $2$-algebra.
\end{thm}
\begin{proof}
  It follows from Lemma \ref{lem:3-form} and we omit details.
\end{proof}

\begin{rmk}
  A semisimple Lie algebra with the Killing form is naturally a quadratic Lie algebra. How to construct a skew-symmetric bilinear form associated to a Leibniz algebra such that it is invariant in the sense of \eqref{Invariant-bilinear-forms} is not known yet.
\end{rmk}

\subsection{Matched pairs, Manin triples of Leibniz algebras and Leibniz bialgebras}
In this subsection, first we recall the notion of a matched pair of Leibniz algebras. Then we introduce the notions of a Manin triple of Leibniz algebras and a Leibniz bialgebra. Finally, we prove the equivalence between matched pairs of Leibniz algebras, Manin triples of Leibniz algebras and Leibniz bialgebras.

\begin{defi}\label{matched-pair}{\rm (\cite{Agore})}
Let $(\g_1,[\cdot,\cdot]_{\g_1})$ and $(\g_2,[\cdot,\cdot]_{\g_2})$ be two Leibniz algebras. If there exists a representation $(\rho^L_{1},\rho^R_{1})$ of $\g_1$ on $\g_2$ and a representation $(\rho^L_{2},\rho^R_{2})$ of $\g_2$ on $\g_1$ such that the identities
\begin{eqnarray}
\label{matched-pair-1}&&\rho^R_{1}(x)[u,v]_{\g_2}-[u,\rho^R_{1}(x)v]_{\g_2}+[v,\rho^R_{1}(x)u]_{\g_2}-\rho^R_{1}(\rho^L_{2}(v)x)u+\rho^R_{1}(\rho^L_{2}(u)x)v=0;\\
\label{matched-pair-2}&&\rho^L_{1}(x)[u,v]_{\g_2}-[\rho^L_{1}(x)u,v]_{\g_2}-[u,\rho^L_{1}(x)v]_{\g_2}-\rho^L_{1}(\rho^R_{2}(u)x)v-\rho^R_{1}(\rho^R_{2}(v)x)u=0;\\
\label{matched-pair-3}&&[\rho^L_{1}(x)u,v]_{\g_2}+\rho^L_{1}(\rho^R_{2}(u)x)v+[\rho^R_{1}(x)u,v]_{\g_2}+\rho^L_{1}(\rho^L_{2}(u)x)v=0;\\
\label{matched-pair-4}&&\rho^R_{2}(u)[x,y]_{\g_1}-[x,\rho^R_{2}(u)y]_{\g_1}+[y,\rho^R_{2}(u)x]_{\g_1}-\rho^R_{2}(\rho^L_{1}(y)u)x+\rho^R_{2}(\rho^L_{1}(x)u)y=0;\\
\label{matched-pair-5}&&\rho^L_{2}(u)[x,y]_{\g_1}-[\rho^L_{2}(u)x,y]_{\g_1}-[x,\rho^L_{2}(u)y]_{\g_1}-\rho^L_{2}(\rho^R_{1}(x)u)y-\rho^R_{2}(\rho^R_{1}(y)u)x=0;\\
\label{matched-pair-6}&&[\rho^L_{2}(u)x,y]_{\g_1}+\rho^L_{2}(\rho^R_{1}(x)u)y+[\rho^R_{2}(u)x,y]_{\g_1}+\rho^L_{2}(\rho^L_{1}(x)u)y=0,
\end{eqnarray}
hold for all $x,y\in\g_1$ and $u,v\in\g_2$, then we call $(\g_1,\g_2;(\rho^L_{1},\rho^R_{1}),(\rho^L_{2},\rho^R_{2}))$ a {\bf matched pair} of Leibniz algebras.

\begin{pro}\label{matched-pair-to-big-algebras}{\rm (\cite{Agore})}
Let $(\g_1,\g_2;(\rho^L_{1},\rho^R_{1}),(\rho^L_{2},\rho^R_{2}))$ be a matched pair of Leibniz algebras. Then there is a Leibniz algebra structure on $\g_1\oplus\g_2$ defined by
\begin{eqnarray}\label{eq:dm}
[x+u,y+v]_{ \bowtie }=[x,y]_{\g_1}+\rho^R_{2}(v)x+\rho^L_{2}(u)y+[u,v]_{\g_2}+\rho^L_{1}(x)v+\rho^R_{1}(y)u.
\end{eqnarray}
\end{pro}

\emptycomment{
\begin{itemize}
        \item[\rm(i)]
        $\rho^R_{1}(x)[u,v]_{\g_2}=[u,\rho^R_{1}(x)v]_{\g_2}-[v,\rho^R_{1}(x)u]_{\g_2}+\rho^R_{1}(\rho^L_{2}(v)x)u-\rho^R_{1}(\rho^L_{2}(u)x)v$;
        \item[\rm(ii)]
        $\rho^L_{1}(x)[u,v]_{\g_2}=[\rho^L_{1}(x)u,v]_{\g_2}+[u,\rho^L_{1}(x)v]_{\g_2}+\rho^L_{1}(\rho^R_{2}(u)x)v+\rho^R_{1}(\rho^R_{2}(v)x)u$;
        \item[\rm(iii)]
        $[\rho^L_{1}(x)u,v]_{\g_2}+\rho^L_{1}(\rho^R_{2}(u)x)v+[\rho^R_{1}(x)u,v]_{\g_2}+\rho^L_{1}(\rho^L_{2}(u)x)v=0$;
        \item[\rm(iv)]
        $\rho^R_{2}(u)[x,y]_{\g_1}=[x,\rho^R_{2}(u)y]_{\g_1}-[y,\rho^R_{2}(u)x]_{\g_1}+\rho^R_{2}(\rho^L_{1}(y)u)x-\rho^R_{2}(\rho^L_{1}(x)u)y$;
        \item[\rm(v)]
        $\rho^L_{2}(u)[x,y]_{\g_1}=[\rho^L_{2}(u)x,y]_{\g_1}+[x,\rho^L_{2}(u)y]_{\g_1}+\rho^L_{2}(\rho^R_{1}(x)u)y+\rho^R_{2}(\rho^R_{1}(y)u)x$;
        \item[\rm(vi)]
        $[\rho^L_{2}(u)x,y]_{\g_1}+\rho^L_{2}(\rho^R_{1}(x)u)y+[\rho^R_{2}(u)x,y]_{\g_1}+\rho^L_{2}(\rho^L_{1}(x)u)y=0$,
\end{itemize}
}
\end{defi}

In the Lie algebra context, to relate  matched pairs of Lie algebras to Lie bialgebras and Manin triples for Lie algebras, we need the notions of the coadjoint representation,  which is the dual representation of the adjoint representation. Now we investigate the dual representation  in the Leibniz algebra context.

\begin{lem}\label{lem:dualrep}
Let $(V;\rho^L,\rho^R)$ be a representation of a Leibniz algebra $(\g,[\cdot,\cdot]_{\g})$. Then $$\big(V^*;(\rho^L)^*,-(\rho^L)^*-(\rho^R)^*\big)$$ is a representation of   $(\g,[\cdot,\cdot]_{\g})$, which is called the {\bf dual representation} of  $(V;\rho^L,\rho^R)$.
\end{lem}
\begin{proof}
By \eqref{rep-1}, for all $x,y\in\g,~v\in V$ and $\chi\in V^*$,  we have
\begin{eqnarray*}
\langle(\rho^L)^*([x,y]_\g)\chi,v\rangle&=&-\langle \chi,\rho^L([x,y]_\g)v\rangle=-\langle \chi,\rho^L(x)\rho^L(y)v-\rho^L(y)\rho^L(x)v\rangle\\
&=&-\langle (\rho^L)^*(y)(\rho^L)^*(x)\chi,v\rangle+\langle (\rho^L)^*(x)(\rho^L)^*(y)\chi,v\rangle
=\langle [(\rho^L)^*(x),(\rho^L)^*(y)]\chi,v\rangle.
\end{eqnarray*}
Thus, we have $(\rho^L)^*([x,y]_\g)=[(\rho^L)^*(x),(\rho^L)^*(y)]$. By \eqref{rep-1} and \eqref{rep-2}, we have
\begin{eqnarray*}
&&\langle \big(-(\rho^L)^*([x,y]_\g)-(\rho^R)^*([x,y]_\g)\big)\chi,v\rangle\\&=&\langle \chi,\rho^L([x,y]_\g)v+\rho^R([x,y]_\g)v\rangle\\
                                                              &=&\langle \chi,\rho^L(x)\rho^L(y)v-\rho^L(y)\rho^L(x)v+\rho^L(x)\rho^R(y)v-\rho^R(y)\rho^L(x)v\rangle\\
                                  &=&\langle (\rho^L)^*(y)(\rho^L)^*(x)\chi,v\rangle-(\rho^L)^*(x)(\rho^L)^*(y)\chi,v\rangle\\
                                  &&+\langle (\rho^R)^*(y)(\rho^L)^*(x)\chi,v\rangle-\langle (\rho^L)^*(x)(\rho^R)^*(y)\chi,v\rangle\\
                                  &=&\langle [(\rho^L)^*(x),-(\rho^L)^*(y)-(\rho^R)^*(y)]\chi,v\rangle.
\end{eqnarray*}
Thus, we have $-(\rho^L)^*([x,y]_\g)-(\rho^R)^*([x,y]_\g)=[(\rho^L)^*(x),-(\rho^L)^*(y)-(\rho^R)^*(y)]$. \emptycomment{ For all $x,y\in\g,v\in V$ and $\chi\in V^*$, by \eqref{rep-3}, we have
\begin{eqnarray*}
&&\langle (-(\rho^L)^*(y)-(\rho^R)^*(y))((\rho^L)^*(x)\chi),v\rangle\\&=&-\langle \chi,(\rho^L(x)\circ \rho^L(y)+\rho^L(x)\circ \rho^R(y))v\rangle\\
                                  &=&-\langle \chi,(\rho^L(x)\circ \rho^L(y)+\rho^R(x)\circ\rho^L(y)+\rho^L(x)\circ \rho^R(y)+\rho^R(x)\circ \rho^R(y))v\rangle\\
                                  &=&-\langle \chi,(\rho^L(x)+\rho^R(x))(\rho^L(y)+\rho^R(y))v\rangle\\
                                  &=&-\langle (-(\rho^L)^*(y)-(\rho^R)^*(y))(-(\rho^L)^*(x)-(\rho^R)^*(x))\chi,v\rangle.
\end{eqnarray*}
}
Similarly, we can show that $$\Big(-(\rho^L)^*(y)-(\rho^R)^*(y)\Big)\circ (\rho^L)^*(x)=-\Big(-(\rho^L)^*(y)-(\rho^R)^*(y)\Big)\circ \Big(-(\rho^L)^*(x)-(\rho^R)^*(x)\Big).$$
Thus, $\big(V^*;(\rho^L)^*,-(\rho^L)^*-(\rho^R)^*\big)$ is a representation of   $(\g,[\cdot,\cdot]_{\g})$.
\end{proof}

\begin{defi}
  A {\bf Manin triple of Leibniz algebras} is a triple $(\huaG,\g_1,\g_2)$, where
  \begin{itemize}
    \item $(\huaG,[\cdot,\cdot]_{\huaG},\omega)$ is a quadratic Leibniz algebra;
    \item both $\g_1$ and $\g_2$ are isotropic subalgebras of $(\huaG,[\cdot,\cdot]_{\huaG})$;
    \item $\huaG=\g_1\oplus\g_2$ as vector spaces.
  \end{itemize}
\end{defi}

\begin{ex}{\rm
 Let $(\g,[\cdot,\cdot]_\g)$ be a Leibniz algebra. Then   $(\g\ltimes_{L^*,-L^*-R^*}\g^*,\g,\g^*)$ is a Manin triple of Leibniz algebras, where the natural nondegenerate skew-symmetric bilinear form $\omega$ on $\g\oplus\g^*$ is given by:
\begin{eqnarray}\label{phase-space}
\omega(x+\xi,y+\eta)=\langle \xi,y\rangle-\langle \eta,x\rangle,\,\,\,\,\forall x,y\in \g,~\xi,\eta\in \g^*.
\end{eqnarray}

 }
\end{ex}

For a Leibniz algebra $(\g^*,[\cdot,\cdot]_{\g^*})$, let $\triangle:\g\longrightarrow\otimes^2 \g$  be the dual map of  $[\cdot,\cdot]_{\g^*}:\otimes^2 \g^*\longrightarrow\g^*$, i.e.
\begin{eqnarray*}
 \langle \triangle x,\xi\otimes \eta\rangle=\langle x,[\xi,\eta]_{\g^*}\rangle.
\end{eqnarray*}

\begin{defi}\label{defi:Leibnizbialgebra}
Let $(\g,[\cdot,\cdot]_\g)$ and $(\g^*,[\cdot,\cdot]_{\g^*})$ be   Leibniz algebras. Then $(\g,\g^*)$ is called a {\bf Leibniz bialgebra} if the following conditions hold:
\begin{itemize}
    \item[\rm(a)]  For all $x,y\in\g$, we have
   $$ \tau_{12}\Big(( R_y\otimes{\Id})(\triangle x)\Big)=(R_x\otimes {\Id})(\triangle y),$$
   where  $\tau_{12}:\g\otimes\g\lon\g\otimes\g$ is the exchange operator defined by
$
\tau_{12}(x\otimes y)=y\otimes x.
$
    \item[\rm(b)] For all $x,y\in\g$, we have
$$\triangle[x,y]_\g=\Big(({\Id}\otimes R_y-L_y\otimes{\Id}-R_y\otimes{\Id})\circ({\Id}+\tau_{12})\Big)\triangle x+\big({\Id}\otimes L_x+L_x\otimes{\Id}\big)\triangle y.$$
  \end{itemize}
\end{defi}

Until now, we have recalled the notion of a matched pair, introduced the notions of a Manin triple of Leibniz algebras and a Leibniz bialgebra. Similar to the case of Lie algebras, these  objects are equivalent when we consider the dual representation of the regular representation in a matched pair of Leibniz algebras.   The following theorem is the main result in this section.

\begin{thm}\label{thm:equivalent}
Let $(\g,[\cdot,\cdot]_\g)$ and $(\g^*,[\cdot,\cdot]_{\g^*})$ be two Leibniz algebras. Then the following conditions are equivalent.
\begin{itemize}
    \item[\rm(i)] $(\g,\g^*)$ is a Leibniz bialgebra.
    \item[\rm(ii)]$(\g,\g^*;(L^*,-L^*-R^*),(\huaL^*,-\huaL^*-\huaR^*))$ is a matched pair of Leibniz algebras.
    \item[\rm(iii)]  $(\g\oplus\g^*,\g,\g^*)$ is a Manin triple of Leibniz algebras, where the invariant skew-symmetric bilinear form on $\g\oplus\g^*$ is given by \eqref{phase-space}.
  \end{itemize}
\end{thm}

\begin{proof}
First we prove that (ii) is  equivalent to (iii).

Let $(\g,\g^*;(L^*,-L^*-R^*),(\huaL^*,-\huaL^*-\huaR^*))$ be a matched pair of Leibniz algebras. Then $(\g\oplus\g^*,[\cdot,\cdot]_{\bowtie})$ is a Leibniz algebra, where $[\cdot,\cdot]_{ \bowtie }$ is given by \eqref{eq:dm}. We only need to prove that $\omega$ satisfies the invariant condition \eqref{Invariant-bilinear-forms}. For all $x,y,z\in\g$ and $\xi,\eta,\alpha\in\g^*$, we have
\begin{eqnarray*}
&&\omega(x+\xi,[y+\eta,z+\alpha]_{\bowtie})\\
&=&\omega(x+\xi,[y,z]_\g+L^*_y\alpha+(-\huaL^*_{\alpha}-\huaR^*_{\alpha})y+\huaL^*_{\eta}z+(-L^*_z-R^*_z)\eta+[\eta,\alpha]_{\g^*})\\
 &=&\langle \xi,[y,z]_\g\rangle-\langle \xi,\huaL^*_{\alpha}y\rangle-\langle \xi,\huaR^*_{\alpha}y\rangle+\langle \xi,\huaL^*_{\eta}z\rangle-\langle L^*_y\alpha,x\rangle+\langle L^*_z\eta,x\rangle+\langle R^*_z\eta,x\rangle-\langle [\eta,\alpha]_{\g^*},x\rangle\\
  &=&\langle \xi,[y,z]_\g\rangle+\langle [\alpha,\xi]_{\g^*},y\rangle+\langle [\xi,\alpha]_{\g^*},y\rangle-\langle [\eta,\xi]_{\g^*},z\rangle\\
                               &&+\langle \alpha,[y,x]_\g\rangle-\langle \eta,[z,x]_\g\rangle-\langle \eta,[x,z]_\g\rangle-\langle [\eta,\alpha]_{\g^*},x\rangle.
\end{eqnarray*}
Moreover, we have
\begin{eqnarray*}
&&\omega([x+\xi,z+\alpha]_{\bowtie}+[z+\alpha,x+\xi]_{\bowtie},y+\eta)\\
 &=&\omega([x,z]_\g+L^*_x\alpha+(-\huaL^*_{\alpha}-\huaR^*_{\alpha})x+\huaL^*_{\xi}z+(-L^*_z-R^*_z)\xi+[\xi,\alpha]_{\g^*}\\
 &&+[z,x]_\g+L^*_z\xi+(-\huaL^*_{\xi}-\huaR^*_{\xi})z+\huaL^*_{\alpha}x+(-L^*_x-R^*_x)\alpha+[\alpha,\xi]_{\g^*},y+\eta)\\
 &=&\omega([x,z]_\g-\huaR^*_{\alpha}x-R^*_z\xi+[\xi,\alpha]_{\g^*}+[z,x]_\g-\huaR^*_{\xi}z-R^*_x\alpha+[\alpha,\xi]_{\g^*},y+\eta)\\
  &=&-\langle R^*_z\xi,y\rangle+\langle [\xi,\alpha]_{\g^*},y\rangle-\langle R^*_x\alpha,y\rangle+\langle [\alpha,\xi]_{\g^*},y\rangle\\
  &&-\langle \eta,[x,z]_\g\rangle+\langle \eta,\huaR^*_{\alpha}x\rangle-\langle \eta,[z,x]_\g\rangle+\langle \eta,\huaR^*_{\xi}z\rangle\\
   &=&\langle \xi,[y,z]_\g\rangle+\langle [\xi,\alpha]_{\g^*},y\rangle+\langle\alpha,[y,x]_\g\rangle+\langle [\alpha,\xi]_{\g^*},y\rangle\\
   &&-\langle \eta,[x,z]_\g\rangle-\langle [\eta,\alpha]_{\g^*},x\rangle-\langle \eta,[z,x]_\g\rangle-\langle [\eta,\xi]_{\g^*},z\rangle.
\end{eqnarray*}
Thus, $\omega$ satisfies the invariant condition \eqref{Invariant-bilinear-forms}.

On the other hand, if $(\huaG,\g,\g^*)$ is a Manin triple of Leibniz algebras with the invariant bilinear form given by \eqref{phase-space}. For all $x\in\g,~\xi,\eta\in\g^*$, by \eqref{Invariant-bilinear-forms}, we have
\begin{eqnarray*}
\langle \eta,\rho_2^R(\xi)x\rangle=\omega(\eta,[x,\xi]_\huaG)=\omega([\eta,\xi]_{\g^*}+[\xi,\eta]_{\g^*},x)=\langle \huaR_{\xi}\eta+\huaL_{\xi}\eta,x\rangle=-\langle \eta,\huaR_{\xi}^*x+\huaL_{\xi}^*x\rangle,
\end{eqnarray*}
which implies that $\rho_2^R=-\huaL^*-\huaR^*$. We have
\begin{eqnarray*}
\langle \rho_1^L(x)\xi,y\rangle=-\omega(y,[x,\xi]_\huaG)=-\omega([y,\xi]_\huaG+[\xi,y]_\huaG,x)=-\omega(\xi,[x,y]_\g)=-\langle\xi,L_xy\rangle=\langle L_x^*\xi,y\rangle,
\end{eqnarray*}
which implies that $\rho_1^L=L^*$. Similarly, we have $\rho_1^R=-L^*-R^*$ and $\rho_2^L=\huaL^*$. Thus, $(\g,\g^*;(L^*,-L^*-R^*),(\huaL^*,-\huaL^*-\huaR^*))$ is a matched pair.  \vspace{3mm}

Next we prove that (i) is  equivalent to (ii).

Let $(\g,[\cdot,\cdot]_\g)$ and $(\g^*,[\cdot,\cdot]_{\g^*})$ be  Leibniz algebras. Consider their representations $(\g^*;L^*,-L^*-R^*)$ and $(\g;\huaL^*,-\huaL^*-\huaR^*)$.   For all $x,y\in\g,~\xi,\eta\in\g^*$, consider the left hand side of \eqref{matched-pair-6}, we have
\begin{eqnarray*}
[\huaL^*_\xi x,y]_\g+\huaL^*_{(-L^*_x-R^*_x)\xi}y+[(-\huaL^*_\xi-\huaR^*_\xi)x,y]_{\g}+\huaL^*_{L^*_x\xi}y=-\huaL^*_{R^*_x\xi}y-[\huaR^*_\xi x,y]_{\g}.
\end{eqnarray*}
Furthermore, by straightforward computations, we have
\begin{eqnarray*}
\langle-\huaL^*_{R^*_x\xi}y-[\huaR^*_\xi x,y]_{\g},\eta\rangle&=&\langle y,[R^*_x\xi,\eta]_{\g^*}\rangle-\langle R_y\huaR^*_\xi x,\eta\rangle=\langle \triangle y,R^*_x\xi\otimes\eta\rangle-\langle \triangle x,R^*_y\eta\otimes\xi \rangle\\
                          &=&-\langle (R_x\otimes{\Id})(\triangle y),\xi\otimes\eta\rangle+\langle (R_y\otimes{\Id})(\triangle x),\eta\otimes\xi \rangle\\
             &=&-\langle (R_x\otimes{\Id})(\triangle y),\xi\otimes\eta\rangle+\langle \tau_{12}[(R_y\otimes{\Id})(\triangle x)],\xi\otimes\eta \rangle.
\end{eqnarray*}
Therefore, \eqref{matched-pair-6} is  equivalent to
\begin{eqnarray}\label{leibniz-bialgebra-1}
(R_x\otimes{\Id})(\triangle y)=\tau_{12}\Big((R_y\otimes{\Id})(\triangle x)\Big).
\end{eqnarray}
The left hand side of  \eqref{matched-pair-5} is equal to
\begin{eqnarray*}
 \huaL^*_\xi[x,y]_\g-[\huaL^*_\xi x,y]_\g-[x,\huaL^*_\xi y]_\g-\huaL^*_{(-L^*_x-R^*_x)\xi}y-(-\huaL^*_{(-L^*_y-R^*_y)\xi}-\huaR^*_{(-L^*_y-R^*_y)\xi})x.
\end{eqnarray*}
Furthermore, by straightforward computations, we have
\begin{eqnarray*}
&&\langle\huaL^*_\xi[x,y]_\g-[\huaL^*_\xi x,y]_\g-[x,\huaL^*_\xi y]_\g+\huaL^*_{L^*_x\xi}y+\huaL^*_{R^*_x\xi}y-\huaL^*_{L^*_y\xi}x-\huaL^*_{R^*_y\xi}x-\huaR^*_{L^*_y\xi}x-\huaR^*_{R^*_y\xi}x,\eta\rangle\\
&=&-\langle[x,y]_\g,[\xi,\eta]_{\g^*}\rangle-\langle x,[\xi ,R^*_y\eta]_{\g^*}\rangle-\langle y,[\xi,L^*_x\eta]_{\g^*}\rangle-\langle y,[L^*_x\xi,\eta]_{\g^*}\rangle-\langle y,[R^*_x\xi,\eta]_{\g^*}\rangle\\
&&+\langle x,[L^*_y\xi,\eta]_{\g^*}\rangle+\langle x,[R^*_y\xi,\eta]_{\g^*}\rangle+\langle x,[\eta,L^*_y\xi]_{\g^*}\rangle+\langle x,[\eta,R^*_y\xi]_{\g^*}\rangle\\
&=&-\langle\triangle[x,y]_\g,\xi\otimes\eta\rangle-\langle\triangle x,\xi\otimes R^*_y\eta\rangle-\langle\triangle y,\xi\otimes L^*_x\eta\rangle-\langle\triangle y,L^*_x\xi\otimes\eta\rangle-\langle\triangle y,R^*_x\xi\otimes\eta\rangle\\
&&+\langle\triangle x,L^*_y\xi\otimes\eta\rangle+\langle\triangle x,R^*_y\xi\otimes\eta\rangle+\langle \triangle x,\eta\otimes L^*_y\xi\rangle+\langle\triangle x,\eta\otimes R^*_y\xi\rangle\\
&=&-\langle\triangle[x,y]_\g,\xi\otimes\eta\rangle+\langle({\Id}\otimes R_y)(\triangle x),\xi\otimes\eta\rangle+\langle({\Id}\otimes L_x)(\triangle y),\xi\otimes\eta\rangle+\langle(L_x\otimes{\Id})(\triangle y),\xi\otimes\eta\rangle\\
&&+\langle(R_x\otimes{\Id})(\triangle y),\xi\otimes\eta\rangle-\langle(L_y\otimes{\Id})(\triangle x),\xi\otimes\eta\rangle-\langle(R_y\otimes{\Id})(\triangle x),\xi\otimes\eta\rangle\\
&&-\langle (\tau_{12}\big(({\Id}\otimes L_y)(\triangle x)\big),\xi\otimes\eta\rangle-\langle\tau_{12}\big(({\Id}\otimes R_y)(\triangle x)\big),\xi\otimes\eta\rangle.
\end{eqnarray*}
Therefore,  \eqref{matched-pair-5} is  equivalent to
\begin{eqnarray}\label{leibniz-bialgebra-2}
\triangle[x,y]_\g&=&\Big({\Id}\otimes R_y-L_y\otimes{\Id}-R_y\otimes{\Id}-\tau_{12}\circ({\Id}\otimes L_y)-\tau_{12}\circ({\Id}\otimes R_y)\Big)(\triangle x)\\
\nonumber&&+\big({\Id}\otimes L_x+L_x\otimes{\Id}+R_x\otimes{\Id}\big)(\triangle y).
\end{eqnarray}
The left hand side of \eqref{matched-pair-4} is equal to
\begin{eqnarray*}
 (-\huaL^*_\xi-\huaR^*_\xi)[x,y]_\g+[x,\huaL^*_\xi y+\huaR^*_\xi y]_\g -[y,\huaL^*_\xi x+\huaR^*_\xi x]_\g +\huaL^*_{L^*_y\xi}x+\huaR^*_{L^*_y\xi}x-\huaL^*_{L^*_x\xi}y-\huaR^*_{L^*_x\xi}y.
\end{eqnarray*}
Furthermore, by straightforward computations, we have
\begin{eqnarray*}
&&\langle(-\huaL^*_\xi-\huaR^*_\xi)[x,y]_\g+[x,\huaL^*_\xi y+\huaR^*_\xi y]_\g -[y,\huaL^*_\xi x+\huaR^*_\xi x]_\g +\huaL^*_{L^*_y\xi}x+\huaR^*_{L^*_y\xi}x-\huaL^*_{L^*_x\xi}y-\huaR^*_{L^*_x\xi}y,\eta\rangle\\
&=&\langle[x,y]_\g,[\xi,\eta]_{\g^*}\rangle+\langle[x,y]_\g,[\eta,\xi]_{\g^*}\rangle+\langle y,[\xi,L^*_x\eta]_{\g^*}\rangle+\langle y,[L^*_x\eta,\xi]_{\g^*}\rangle-\langle x,[\xi,L^*_y\eta]_{\g^*}\rangle\\
&&-\langle x,[L^*_y\eta,\xi]_{\g^*}\rangle-\langle x,[L^*_y\xi,\eta]_{\g^*}\rangle-\langle x,[\eta,L^*_y\xi]_{\g^*}\rangle+
\langle y,[L^*_x\xi,\eta]_{\g^*}\rangle+\langle y,[\eta,L^*_x\xi]_{\g^*}\rangle\\
&=&\langle\triangle[x,y]_\g,\xi\otimes\eta\rangle+\langle\tau_{12}(\triangle[x,y]_\g),\xi\otimes\eta\rangle-\langle({\Id}\otimes L_x)(\triangle y),\xi\otimes\eta\rangle-\langle\tau_{12}\big((L_x\otimes{\Id})(\triangle y)\big),\xi\otimes\eta\rangle\\
&&+\langle({\Id}\otimes L_y)\triangle x,\xi\otimes\eta\rangle+\langle\tau_{12}\big((L_y\otimes{\Id})(\triangle x)\big),\xi\otimes\eta\rangle+\langle(L_y\otimes{\Id})(\triangle x),\xi\otimes\eta\rangle\\
&&+\langle\tau_{12}\big(({\Id}\otimes L_y)(\triangle x)\big),\xi\otimes\eta\rangle-\langle(L_x\otimes{\Id})(\triangle y),\xi\otimes\eta\rangle-\langle\tau_{12}\big(({\Id}\otimes L_x)(\triangle y)\big),\xi\otimes\eta\rangle.
\end{eqnarray*}
Therefore, \eqref{matched-pair-4} is  equivalent to
\begin{eqnarray}
\label{leibniz-bialgebra-3}\triangle[x,y]_\g+\tau_{12}(\triangle[x,y]_\g)&=&\big({\Id}\otimes L_x+\tau_{12}\circ(L_x\otimes{\Id})+L_x\otimes{\Id}+\tau_{12}\circ({\Id}\otimes L_x)\big)(\triangle y)\\
                                              \nonumber&&-\big({\Id}\otimes L_y+\tau_{12}\circ(L_y\otimes{\Id})+L_y\otimes{\Id}+\tau_{12}\circ({\Id}\otimes L_y)\big)(\triangle x).
\end{eqnarray}
By \eqref{leibniz-bialgebra-1} and \eqref{leibniz-bialgebra-2}, we deduce that
\begin{eqnarray*}
&&\triangle[x,y]_\g+\tau_{12}(\triangle[x,y]_\g)\\&=&\big(\underline{{\Id}\otimes R_y}-L_y\otimes{\Id}\underbrace{-R_y\otimes{\Id}}-\tau_{12}\circ({\Id}\otimes L_y)\underline{-\tau_{12}\circ({\Id}\otimes R_y)}\big)(\triangle x)\\
&&+\big({\Id}\otimes L_x+L_x\otimes{\Id}+\underbrace{R_x\otimes{\Id}}\big)(\triangle y)\\
&&+\big(\underline{\tau_{12}\circ({\Id}\otimes R_y)}-\tau_{12}\circ(L_y\otimes{\Id})\underbrace{-\tau_{12}\circ(R_y\otimes{\Id})}-{\Id}\otimes L_y\underline{-{\Id}\otimes R_y}\big)(\triangle x)\\
&&+\big(\tau_{12}\circ({\Id}\otimes L_x)+\tau_{12}\circ(L_x\otimes{\Id})+\underbrace{\tau_{12}\circ(R_x\otimes{\Id})}\big)(\triangle y)\\
                                              &=& \mbox{the right hand side of }\eqref{leibniz-bialgebra-3}.
\end{eqnarray*}
Thus, by \eqref{matched-pair-5} and \eqref{matched-pair-6}, we can deduce that \eqref{matched-pair-4} holds.

Consider the left hand side of \eqref{matched-pair-3}, it equals to $-L^*_{\huaR^*_\xi x}\eta-[R^*_x\xi,\eta]_{\g^*}$. For all $y\in\g$, we have
\begin{eqnarray*}
\langle-L^*_{\huaR^*_\xi x}\eta-[R^*_x\xi,\eta]_{\g^*},y\rangle=\langle \eta,[\huaR^*_\xi x,y]_\g\rangle+\langle \eta,\huaL^*_{R^*_x\xi}y\rangle=\langle \eta,[\huaR^*_\xi x,y]_\g+\huaL^*_{R^*_x\xi}y\rangle,
\end{eqnarray*}
which implies that \eqref{matched-pair-3} is equivalent to \eqref{matched-pair-6}.
Similarly, if \eqref{matched-pair-3} holds, we can deduce that \eqref{matched-pair-2} is equivalent to \eqref{matched-pair-5}. Furthermore, by \eqref{matched-pair-2} and \eqref{matched-pair-3}, we  can deduce that \eqref{matched-pair-1} holds naturally. Therefore, $(\g,\g^*;(L^*,-L^*-R^*),(\huaL^*,-\huaL^*-\huaR^*))$ is a matched pair of Leibniz algebras  if and only if  \eqref{leibniz-bialgebra-1} and  \eqref{leibniz-bialgebra-2} hold. Note that \eqref{leibniz-bialgebra-1} is exactly Condition (a) in Definition \ref{defi:Leibnizbialgebra}. Furthermore, if \eqref{leibniz-bialgebra-1} holds, \eqref{leibniz-bialgebra-2} is exactly Condition (b) in Definition \ref{defi:Leibnizbialgebra}. Thus, $(\g,\g^*)$ is a Leibniz bialgebra
   if and only if $(\g,\g^*;(L^*,-L^*-R^*),(\huaL^*,-\huaL^*-\huaR^*))$ is a matched pair of Leibniz algebras. The proof is finished.
\end{proof}

\begin{cor}
  Let $(\g,\g^*)$ be a Leibniz bialgebra. Then $(\g^*,\g)$ is also a Leibniz bialgebra.
\end{cor}

\section{(Relative) Rota-Baxter operators and  twisting theory}\label{sec:K}

In this section, first we recall the graded Lie algebra whose Maurer-Cartan elements are Leibniz algebra structures, and define the bidegree of a multilinear map which is the technical tool in our later study. Then we introduce the notion of a relative Rota-Baxter operator on a Leibniz algebra, and construct the graded Lie algebra whose Maurer-Cartan elements are relative Rota-Baxter operators.  Finally, we give the twisting theory of twilled Leibniz algebras. These structures and theories are the main ingredient in our later study of Leibniz bialgebras.

\subsection{Lift and bidegree of multilinear maps}\label{sec:K1}

A permutation $\sigma\in\mathbb S_n$ is called an $(i,n-i)$-{\bf shuffle} if $\sigma(1)<\cdots <\sigma(i)$ and $\sigma(i+1)<\cdots <\sigma(n)$. If $i=0$ or $n$ we assume $\sigma=\Id$. The set of all $(i,n-i)$-shuffles will be denoted by $\mathbb S_{(i,n-i)}$. The notion of an $(i_1,\cdots,i_k)$-shuffle and the set $\mathbb S_{(i_1,\cdots,i_k)}$ are defined analogously.

Let $\g$ be a vector space. We consider the graded vector space $C^*(\g,\g)=\oplus_{n\ge 1}C^n(\g,\g)=\oplus_{n\ge 1}\Hom(\otimes^n\g,\g)$. The {\bf Balavoine bracket} on the graded vector space $C^*(\g,\g)$ is given
by:
\begin{eqnarray}\label{leibniz-bracket}
[P,Q]_\B=P\bar{\circ}Q-(-1)^{pq}Q\bar{\circ}P,\,\,\,\,\forall P\in C^{p+1}(\g,\g),Q\in C^{q+1}(\g,\g),
\end{eqnarray}
where $P\bar{\circ}Q\in C^{p+q+1}(\g,\g)$ is defined by
\begin{eqnarray}
P\bar{\circ}Q=\sum_{k=1}^{p+1}(-1)^{(k-1)q}P\circ_k Q,
\end{eqnarray}
and $\circ_k$ is defined by
\begin{eqnarray}
 \nonumber&&(P\circ_kQ)(x_1,\cdots,x_{p+q+1})\\
&=&\sum_{\sigma\in\mathbb S_{(k-1,q)}}(-1)^{\sigma}P(x_{\sigma(1)},\cdots,x_{\sigma(k-1)},Q(x_{\sigma(k)},\cdots,x_{\sigma(k+q-1)},x_{k+q}),x_{k+q+1},\cdots,x_{p+q+1}).
\end{eqnarray}
It is well known that

\begin{thm}{\rm (\cite{Balavoine-1,Fialowski})}\label{leibniz-algebra-B}
With the above notations, $(C^*(\g,\g),[\cdot,\cdot]_{\B})$ is a graded Lie algebra. Its Maurer-Cartan elements are precisely the Leibniz algebra structures on $\g$.
\end{thm}

Let $\g_1$ and $\g_2$ be vector spaces and elements in $\g_1$ will be denoted by $x,y,z, x_i$ and elements in $\g_2$ will be denoted by $u,v,v_i$. Let $c:\g_2^{\otimes n}\lon \g_1$ be a linear map. We can construct a linear map $\hat{c}\in C^n(\g_1\oplus\g_2,\g_1\oplus\g_2)$ by
\begin{eqnarray*}
\hat{c}\big((x_1,v_1)\otimes\cdots\otimes(x_n,v_n)\big):=(c(v_1,\cdots,v_n),0).
\end{eqnarray*}
In general, for a given linear map $f:\g_{i(1)}\otimes\g_{i(2)}\otimes\cdots\otimes\g_{i(n)}\lon\g_j$, $i(1),\cdots,i(n),j\in\{1,2\}$, we define a linear map $\hat{f}\in C^n(\g_1\oplus\g_2,\g_1\oplus\g_2)$ by

\[
\hat{f}:=\left\{
\begin{array}{ll}
f &\mbox {on $\g_{i(1)}\otimes\g_{i(2)}\otimes\cdots\otimes\g_{i(n)}$, }\\
0 &\mbox {all other cases.}
\end{array}
\right.
\]
We call the linear map $\hat{f}$ a {\bf horizontal lift} of $f$, or simply a lift. Let $H:\g_2\lon\g_1$ be a linear map. Its lift is given by
$
\hat{H}(x,v)=(H(v),0).
$
Obviously we have $\hat{H}\circ\hat{H}=0.$\vspace{2mm}

We denote by $\g^{l,k}$ the direct sum of all $(l+k)$-tensor powers of $\g_1$ and $\g_2$, where $l$ (resp. $k$) is the number of $\g_1$ (resp. $\g_2$).
By the properties of the $\Hom$-functor, we have
\begin{eqnarray}\label{decomposition}
C^n(\g_1\oplus\g_2,\g_1\oplus\g_2)\cong\sum_{l+k=n}\Hom(\g^{l,k},\g_1)\oplus\sum_{l+k=n}\Hom(\g^{l,k},\g_2),
\end{eqnarray}
where the isomorphism is the horizontal lift.

\begin{defi}
A linear map $f\in \Hom\big(\otimes^n(\g_1\oplus\g_2),(\g_1\oplus\g_2)\big)$ has a {\bf bidegree} $l|k$, which is denoted by $||f||=l|k$,   if $f$ satisfies the following four conditions:
\begin{itemize}
\item[\rm(i)] $l+k+1=n;$
\item[\rm(ii)] If $X$ is an element in $\g^{l+1,k}$, then $f(X)\in\g_1;$
\item[\rm(iii)] If $X$ is an element in $\g^{l,k+1}$, then $f(X)\in\g_2;$
\item[\rm(iv)] All the other case, $f(X)=0.$
\end{itemize}
\end{defi}
A linear map $f$ is said to be homogeneous  if $f$ has a bidegree.
We have $l+k\ge0,~k,l\ge-1$ because $n\ge1$ and $l+1,~k+1\ge0$. For instance, the lift $\hat{H}\in C^1(\g_1\oplus\g_2,\g_1\oplus\g_2)$ of $H:\g_2\lon\g_1$ has the bidegree $-1|1$.

It is obvious that we have the following lemmas:
\emptycomment{
\begin{lem}
  The bidegree of $f\in C^n(\g_1\oplus\g_2,\g_1\oplus\g_2)$ is $l|k$ if and only if the following four conditions hold:
\begin{itemize}
\item[\rm(i)] $l+k+1=n;$
\item[\rm(ii)] If $X$ is an element in $\g^{l+1,k}$, then $f(X)\in\g_1;$
\item[\rm(iii)] If $X$ is an element in $\g^{l,k+1}$, then $f(X)\in\g_2;$
\item[\rm(iv)] All the other case, $f(X)=0.$
\end{itemize}
\end{lem}
}

\begin{lem}\label{Zero-condition-1}
Let $f_1,\cdots,f_k\in C^n(\g_1\oplus\g_2,\g_1\oplus\g_2)$ be homogeneous linear maps and the bidegrees of $f_i$ are
different. Then $f_1+\cdots+f_k=0$ if and only if $f_1=\cdots=f_k=0.$
\end{lem}

\begin{lem}\label{Zero-condition-2}
If $||f||=-1|l$ (resp. $l|-1$) and $||g||=-1|k$ (resp. $k|-1$), then $[f,g]_\B=0.$
\end{lem}
\begin{proof}
Assume that $||f||=-1|l$ and $||g||=-1|k$. Then $f$ and $g$ are both horizontal lift of linear maps in $C^*(\g_2,\g_1)$. By the definition of the lift, we have $f\circ_i g=g\circ_j f=0$ for any $i,j.$ Thus, we have $[f,g]_\B=0.$
\end{proof}

\begin{lem}\label{important-lemma-1}
Let $f\in C^{n}(\g_1\oplus\g_2,\g_1\oplus\g_2)$ and $g\in C^{m}(\g_1\oplus\g_2,\g_1\oplus\g_2)$ be homogeneous linear maps with bidegrees $l_f|k_f$ and $l_g|k_g$ respectively. Then the composition $f\circ_ig\in C^{n+m-1}(\g_1\oplus\g_2,\g_1\oplus\g_2)$ is  a homogeneous linear map of the bidegree $l_f+l_g|k_f+k_g.$
\end{lem}

\begin{lem}\label{important-lemma-2}
If $||f||=l_f|k_f$ and $||g||=l_g|k_g$, then $[f,g]_\B$ has the bidegree $l_f+l_g|k_f+k_g.$
\end{lem}
\begin{proof}
By Lemma \ref{important-lemma-1} and \eqref{leibniz-bracket}, we have $||[f,g]_\B||=l_f+l_g|k_f+k_g.$
\end{proof}

\subsection{(Relative) Rota-Baxter operators}\label{sec:K2}
First we introduce the notion of a (relative) Rota-Baxter operator and give some examples.

\begin{defi} \label{defi:O} Let $(\g,[\cdot,\cdot]_\g)$ be a Leibniz algebra.
\begin{enumerate}
\item[\rm(i)] A linear operator $R:\g\longrightarrow \g$ is called a {\bf Rota-Baxter operator } if
\begin{equation} [R(x),R(y)]_\g=R\big([R(x),y]_\g+ [x,R(y)]_\g \big), \quad \forall x, y \in \g.
\label{eq:rbo}
\end{equation}
\item[\rm(ii)]
Let $(V;\rho^L,\rho^R)$ be a representation of a Leibniz algebra $(\g,[\cdot,\cdot]_\g)$. A {\bf relative Rota-Baxter operator} on $\g$ with respect to the representation $(V;\rho^L,\rho^R)$ is a linear map $K:V\longrightarrow\g$ such that
 \begin{eqnarray}\label{O-operator}
[Kv_1,Kv_2]_\g=K(\rho^L(Kv_1)v_2+\rho^R(Kv_2)v_1),\,\,\,\,\forall v_1,v_2\in V.
\end{eqnarray}
\end{enumerate}
\end{defi}

\begin{rmk}
  When $(\g,[\cdot,\cdot]_\g)$ is a Lie algebra and $\rho^R=-\rho^L$, we obtain the notion of a relative Rota-Baxter operator (an $\huaO$-operator) on a Lie algebra with respect to a representation.
\end{rmk}

\begin{ex}\label{example-6}{\rm
Consider the $2$-dimensional Leibniz algebra $(\g,[\cdot,\cdot])$ given with respect to a basis $\{e_1,e_2\}$   by
\begin{eqnarray*}
[e_1,e_1]=0,\quad [e_1,e_2]=0,\quad [e_2,e_1]=e_1,\quad [e_2,e_2]=e_1.
\end{eqnarray*}
Let $\{e_1^*,e_2^*\}$ be the dual basis.
Then $K=\left(\begin{array}{cc}a_{11}&a_{12}\\
                                                                a_{21}&a_{22}\end{array}\right)$ is a relative Rota-Baxter operator on $(\g,[\cdot,\cdot])$ with respect to the representation $(\g^*;L^*,-L^*-R^*)$\footnote{It is the dual representation of the regular representation. See Lemma \ref{lem:dualrep}.} if and only if
$$
   [Ke_i^*,Ke_j^*]=K\Big(L^*_{Ke_i^*}e_j^*-L^*_{Ke_j^*}e_i^*-R^*_{Ke_j^*}e_i^*\Big),\quad \forall i,j=1,2.
   $$
It is straightforward to deduce that
\begin{eqnarray*}
L_{e_1}(e_1,e_2)=(e_1,e_2)\left(\begin{array}{cc}0&0\\
                                                 0&0\end{array}\right),\quad
L_{e_2}(e_1,e_2)=(e_1,e_2)\left(\begin{array}{cc}1&1\\
                                                 0&0\end{array}\right),\\
R_{e_1}(e_1,e_2)=(e_1,e_2)\left(\begin{array}{cc}0&1\\
                                                 0&0\end{array}\right),\quad
R_{e_2}(e_1,e_2)=(e_1,e_2)\left(\begin{array}{cc}0&1\\
                                                 0&0\end{array}\right),
\end{eqnarray*}
and
\begin{eqnarray*}
L_{e_1}^*(e_1^*,e_2^*)=(e_1^*,e_2^*)\left(\begin{array}{cc}0&0\\
                                                           0&0\end{array}\right),\quad
L_{e_2}^*(e_1^*,e_2^*)=(e_1^*,e_2^*)\left(\begin{array}{cc}-1&0\\
                                                           -1&0\end{array}\right),\\
R_{e_1}^*(e_1^*,e_2^*)=(e_1^*,e_2^*)\left(\begin{array}{cc}0&0\\
                                                          -1&0\end{array}\right),\quad
R_{e_2}^*(e_1^*,e_2^*)=(e_1^*,e_2^*)\left(\begin{array}{cc}0&0\\
                                                          -1&0\end{array}\right).
\end{eqnarray*}
 We have
\begin{eqnarray*}
[Ke_1^*,Ke_1^*]=[a_{11}e_1+a_{21}e_2,a_{11}e_1+a_{21}e_2]=a_{21}(a_{11}+a_{21})e_1,
\end{eqnarray*}
and
\begin{eqnarray*}
K\Big(L^*_{Ke_1^*}e_1^*-L^*_{Ke_1^*}e_1^*-R^*_{Ke_1^*}e_1^*\Big)=(a_{11}+a_{21})K(e_2^*)=a_{12}(a_{11}+a_{21})e_1+a_{22}(a_{11}+a_{21})e_2.
\end{eqnarray*}
Thus, by $[Ke_1^*,Ke_1^*]=K\Big(L^*_{Ke_1^*}e_1^*-L^*_{Ke_1^*}e_1^*-R^*_{Ke_1^*}e_1^*\Big)$, we obtain
$$
a_{21}(a_{11}+a_{21})=a_{12}(a_{11}+a_{21}),\quad a_{22}(a_{11}+a_{21})=0.
$$
Similarly, we obtain
 \begin{eqnarray*}
 a_{21}(a_{12}+a_{22})&=&a_{22}a_{11}+(a_{12}+2a_{22})a_{12},\quad a_{22}(a_{21}+a_{12}+2a_{22})=0,\\
 a_{22}(a_{11}+a_{21})&=&-a_{22}(a_{11}+a_{12}),\quad -a_{22}(a_{21}+a_{22})=0,\quad a_{22}(a_{12}+a_{22})=0.
\end{eqnarray*}
 \emptycomment{
By $[Ke_1^*,Ke_2^*]=K\Big(L^*_{Ke_1^*}e_2^*-L^*_{Ke_2^*}e_1^*-R^*_{Ke_2^*}e_1^*\Big)$, we have
\begin{eqnarray*}
[Ke_1^*,Ke_2^*]=[a_{11}e_1+a_{21}e_2,a_{12}e_1+a_{22}e_2]=a_{21}(a_{12}+a_{22})e_1,
\end{eqnarray*}
and
\begin{eqnarray*}
K\Big(L^*_{Ke_1^*}e_2^*-L^*_{Ke_2^*}e_1^*-R^*_{Ke_2^*}e_1^*\Big)&=&K\Big(a_{22}e_1^*+(a_{12}+2a_{22})e_2^*\Big)\\
                                                                &=&a_{22}a_{11}e_1+a_{22}a_{21}e_2+(a_{12}+2a_{22})a_{12}e_1+(a_{12}+2a_{22})a_{22}e_2\\
                                                                &=&\Big(a_{22}a_{11}+(a_{12}+2a_{22})a_{12}\Big)e_1+a_{22}(a_{21}+a_{12}+2a_{22})e_2.
\end{eqnarray*}
-----
By $[Ke_2^*,Ke_1^*]=K\Big(L^*_{Ke_2^*}e_1^*-L^*_{Ke_1^*}e_2^*-R^*_{Ke_1^*}e_2^*\Big)$, we have
\begin{eqnarray*}
[Ke_2^*,Ke_1^*]=[a_{12}e_1+a_{22}e_2,a_{11}e_1+a_{21}e_2]=a_{22}(a_{11}+a_{21})e_1,
\end{eqnarray*}
and
\begin{eqnarray*}
K\Big(L^*_{Ke_2^*}e_1^*-L^*_{Ke_1^*}e_2^*-R^*_{Ke_1^*}e_2^*\Big)&=&-a_{22}K\Big(e_1^*+e_2^*\Big)\\
                                                                &=&-a_{22}a_{11}e_1-a_{22}a_{21}e_2-a_{22}a_{12}e_1-a_{22}a_{22}e_2\\
                                                                &=&-a_{22}(a_{11}+a_{12})e_1-a_{22}(a_{21}+a_{22})e_2.
\end{eqnarray*}
-----
By $[Ke_2^*,Ke_2^*]=K\Big(L^*_{Ke_2^*}e_2^*-L^*_{Ke_2^*}e_2^*-R^*_{Ke_2^*}e_2^*\Big)$, we have
\begin{eqnarray*}
[Ke_2^*,Ke_2^*]=[a_{12}e_1+a_{22}e_2,a_{12}e_1+a_{22}e_2]=a_{22}(a_{12}+a_{22})e_1,
\end{eqnarray*}
and $K\Big(L^*_{Ke_2^*}e_2^*-L^*_{Ke_2^*}e_2^*-R^*_{Ke_2^*}e_2^*\Big)=0.$
}
Summarize the above discussion, we have
\begin{itemize}
     \item[\rm(i)] If $a_{22}=0$, then    $K=\left(\begin{array}{cc}a_{11}&a_{12}\\
   a_{21}&0\end{array}\right)$ is a relative Rota-Baxter operator ~on $(\g,[\cdot,\cdot])$ with respect to the representation $(\g^*;L^*,-L^*-R^*)$ if and only if
   $$
   (a_{12}-a_{21})a_{12}=(a_{12}-a_{21})(a_{11}+a_{21})=0.
   $$
   More precisely, any $K=\left(\begin{array}{cc}a &b\\
b&0\end{array}\right)$ or $K=\left(\begin{array}{cc}a &0\\
-a&0\end{array}\right)$ is a relative Rota-Baxter operator.
     \item[\rm(ii)] If $a_{22}\not=0$, then    $K=\left(\begin{array}{cc}a_{11}&a_{12}\\
   a_{21}&a_{22}\end{array}\right)$ is a relative Rota-Baxter operator~ on $(\g,[\cdot,\cdot] )$ with respect to the representation $(\g^*;L^*,-L^*-R^*)$ if and only if
   $$
   a_{11}=-a_{12}=-a_{21}=a_{22}.
   $$
   \end{itemize}
   }
\end{ex}

In the sequel, we construct the graded Lie algebra that characterize relative Rota-Baxter operators as Maurer-Cartan elements.

Let $(V;\rho^L,\rho^R)$ be  a representation of a Leibniz algebra $(\g,[\cdot,\cdot]_\g)$. Then there is a Leibniz algebra structure on $\g\oplus V$  given by
\begin{eqnarray}
[x+u,y+v]_{\ltimes}=[x,y]_\g+\rho^L(x)v+\rho^R(y)u, \quad \forall x,y\in\g,~ u,v\in V.
\end{eqnarray}
This Leibniz algebra is called the semidirect product of $\g$ and $(V;\rho^L,\rho^R)$, and denoted by $\g\ltimes_{\rho^L,\rho^R}V.$ We denote the above semidirect product Leibniz multiplication by $\hat{\mu}_1.$

Consider the graded vector space
$$C^*(V,\g):=\oplus_{n\ge 1}C^n(V,\g)=\oplus_{n\geq 1}\Hom(\otimes^{n}V,\g).$$

\begin{thm}\label{twilled-DGLA}
With the above notations,  $(C^*(V,\g),\{\cdot,\cdot\})$ is a   graded Lie algebra, where the graded Lie bracket $\{\cdot,\cdot\}:C^m(V,\g)\times C^n(V,\g)\lon C^{m+n}(V,\g)$ is defined by
\begin{eqnarray*}
\{g_1,g_2\}&=&(-1)^{|g_1|}[[\hat{\mu}_1,\hat{g}_1]_\B,\hat{g}_2]_\B,\quad \forall g_1\in C^m(V,\g),~g_2\in C^n(V,\g).
\end{eqnarray*}
 More precisely, we have
{\footnotesize
\begin{eqnarray}
&&\nonumber\{g_1,g_2\}(v_1,v_2,\cdots,v_{m+n})\\
&=&\nonumber\sum_{k=1}^{m}\sum_{\sigma\in\mathbb S_{(k-1,n)}}(-1)^{(k-1)n+1}(-1)^{\sigma}g_1(v_{\sigma(1)},\cdots,v_{\sigma(k-1)},\rho^L(g_2(v_{\sigma(k)},\cdots,v_{\sigma(k+n-1)}))v_{k+n},v_{k+n+1},\cdots,v_{m+n})\\
&&\nonumber+\sum_{k=2}^{m+1}\sum_{\sigma\in\mathbb S_{(k-2,n,1)}\atop \sigma(k+n-2)=k+n-1}(-1)^{kn}(-1)^{\sigma}
g_1(v_{\sigma(1)},\cdots,v_{\sigma(k-2)},\rho^R(g_2(v_{\sigma(k-1)},\cdots,v_{\sigma(k+n-2)}))v_{\sigma(k+n-1)},v_{k+n},\cdots,v_{m+n})\\
&&\nonumber+\sum_{k=1}^{m}\sum_{\sigma\in\mathbb S_{(k-1,n-1)}}(-1)^{(k-1)n}(-1)^{\sigma}[g_2(v_{\sigma(k)},\cdots,v_{\sigma(k+n-2)},v_{k+n-1}),g_1(v_{\sigma(1)},\cdots,v_{\sigma(k-1)},v_{k+n},\cdots,v_{m+n})]_{\g}\\
&&\nonumber+\sum_{\sigma\in\mathbb S_{(m,n-1)}}(-1)^{mn+1}(-1)^{\sigma}[g_1(v_{\sigma(1)},\cdots,v_{\sigma(m)}),g_2(v_{\sigma(m+1)},\cdots,v_{\sigma(m+n-1)},v_{m+n})]_{\g}\\
&&\nonumber+\sum_{k=1}^{n}\sum_{\sigma\in\mathbb S_{(k-1,m)}}(-1)^{m(k+n-1)}(-1)^{\sigma}g_2(v_{\sigma(1)},\cdots,v_{\sigma(k-1)},\rho^L(g_1(v_{\sigma(k)},\cdots,v_{\sigma(k+m-1)}))v_{k+m},v_{k+m+1},\cdots,v_{m+n})\\
&&\nonumber+\sum_{k=1}^{n}\sum_{\sigma\in\mathbb S_{(k-1,m,1)}\atop\sigma(k+m-1)=k+m}(-1)^{m(k+n-1)+1}(-1)^{\sigma}
g_2(v_{\sigma(1)},\cdots,v_{\sigma(k-1)},\rho^R(g_1(v_{\sigma(k)},\cdots,v_{\sigma(k-1+m)}))v_{\sigma(k+m)},v_{k+m+1},\cdots,v_{m+n}).
\end{eqnarray}
}
Moreover, its Maurer-Cartan elements are relative Rota-Baxter operators on the Leibniz algebra $(\g,[\cdot,\cdot]_\g)$ with respect to the representation $(V;\rho^L,\rho^R)$.
\end{thm}

\begin{proof}
The graded Lie algebra $(C^*(V,\g),\{\cdot,\cdot\})$ is obtained via the derived bracket \cite{Kosmann-Schwarzbach,Voronov1}. In fact, the Balavoine bracket $[\cdot,\cdot]_\B$ associated to the direct sum vector space $\g\oplus V$ gives rise to a graded Lie algebra $(C^*(\g\oplus V,\g\oplus V),[\cdot,\cdot]_\B)$. Since $\hat{\mu}_1$ is the semidirect product Leibniz algebra structure on the vector space $\g\oplus V$. By Theorem \ref{leibniz-algebra-B}, we deduce that $(C^*(\g\oplus V,\g\oplus V),[\cdot,\cdot]_\B,d=[\hat{\mu}_1,\cdot]_\B)$ is a differential graded Lie algebra. Obviously $C^*(V,\g)$ is an abelian subalgebra. Further, we define the derived bracket on the graded vector space $C^*(V,\g)$ by
\begin{eqnarray*}
\{g_1,g_2\}:=(-1)^{|g_1|}[d(\hat{g}_1),\hat{g}_2]_\B=(-1)^{|g_1|}[[\hat{\mu}_1,\hat{g}_1]_\B,\hat{g}_2]_\B,\quad \forall g_1\in C^m(V,\g),~g_2\in C^n(V,\g).
\end{eqnarray*}
By Lemma \ref{important-lemma-2}, the derived bracket $\{\cdot,\cdot\}$ is closed on $C^*(V,\g)$, which implies  that $(C^*(V,\g),\{\cdot,\cdot\})$ is a graded Lie algebra. Moreover, it is straightforward to obtain the above concrete graded Lie  bracket $\{\cdot,\cdot\}$ on $C^*(V,\g)=\oplus_{k=1}^{+\infty}C^k(V,\g)$.

For all $K\in C^1(V,\g)$, we have
\begin{eqnarray*}
\{K,K\}(v_1,v_2)=2([Kv_1,Kv_2]_\g-K(\rho^L(Kv_1)v_2)-K(\rho^R(Kv_2)v_1)),\quad\forall v_1,v_2\in V.
\end{eqnarray*}
Thus, Maurer-Cartan elements are precisely relative Rota-Baxter operators on $\g$ with respect to the representation $(V;\rho^L,\rho^R)$. The proof is finished.
\end{proof}

  This is the main ingredient in our later study of the classical Leibniz Yang-Baxter equation and the classical Leibniz $r$-matrix.

\subsection{Twilled Leibniz algebras and the twisting theory}\label{sec:K3}

Let $(\huaG,[\cdot,\cdot]_\huaG)$ be a Leibniz algebra with a decomposition  into two subspaces\footnote{Here $\g_1$ and $\g_2$ are not necessarily subalgebras.}, $\huaG=\g_1\oplus\g_2$.   For later convenience, we use $\Omega$ to denote the multiplication $[\cdot,\cdot]_\huaG$, i.e. $$\Omega((x,u),(y,v)):=[(x,u),(y,v)]_{\huaG}.$$

\begin{lem}\label{lem:dec}
Any $\Omega\in C^2(\huaG,\huaG)$ is uniquely decomposed into four homogeneous linear maps of bidegrees $2|-1,~1|0,~0|1$ and $-1|2,$
$$
\Omega=\hat{\phi}_1+\hat{\mu}_1+\hat{\mu}_2+\hat{\phi}_2.
$$
\end{lem}
\begin{proof}
By \eqref{decomposition},  $C^2(\huaG,\huaG)$ is decomposed into
$$
C^2(\huaG,\huaG)=(2|-1)+(1|0)+(0|1)+(-1|2),
$$
where $(i|j)$ is the space of linear maps  of the bidegree $i|j$.  By Lemma \ref{Zero-condition-1}, $\Omega$ is uniquely decomposed into homogeneous linear maps of bidegrees $2|-1,~1|0,~0|1$ and $-1|2$.
\end{proof}

Denote $P_{\g_i}[X,Y]_\huaG$ by $[X,Y]_i$, for $X,Y\in\huaG,~i=1,2$, where $P_{\g_1}$ and $P_{\g_2}$ are the natural projections from $\huaG$ to $\g_1$ and $\g_2$ respectively. The multiplication $[(x,u),(y,v)]_{\huaG}$ of $\huaG$ is uniquely decomposed by the canonical projections $P_{\g_1}$ and $P_{\g_2}$ into eight multiplications:
\begin{eqnarray*}
\label{1}[x,y]_{\huaG}&=&([x,y]_1,[x,y]_2),\quad
\label{2}[x,v]_{\huaG}=([x,v]_1,[x,v]_2),\\
\label{3}[u,y]_{\huaG}&=&([u,y]_1,[u,y]_2),\quad
\label{4}[u,v]_{\huaG}=([u,v]_1,[u,v]_2).
\end{eqnarray*}
 Write $\Omega=\hat{\phi}_1+\hat{\mu}_1+\hat{\mu}_2+\hat{\phi}_2$ as in Lemma \ref{lem:dec}.  Then we obtain
\begin{eqnarray}
\label{bracket-1}\hat{\phi}_1((x,u),(y,v))&=&(0,[x,y]_2),\\
\label{bracket-2}\hat{\mu}_1((x,u),(y,v))&=&([x,y]_1,[x,v]_2+[u,y]_2),\\
\label{bracket-3}\hat{\mu}_2((x,u),(y,v))&=&([x,v]_1+[u,y]_1,[u,v]_2),\\
\label{bracket-4}\hat{\phi}_2((x,u),(y,v))&=&([u,v]_1,0).
\end{eqnarray}
Observe that $\hat{\phi}_1$ and $\hat{\phi}_2$ are lifted linear maps of $\phi_1(x,y):=[x,y]_2$ and $\phi_2(u,v):=[u,v]_1.$

\begin{defi}
The triple $(\huaG,\g_1,\g_2)$ is called a {\bf twilled Leibniz algebra} if $\phi_1=\phi_2=0$, or equivalently, $\g_1$ and $\g_2$ are subalgebras of $\huaG$.
\end{defi}

\begin{lem}\label{lem:twillL}
The triple $(\huaG,\g_1,\g_2)$ is a twilled Leibniz algebra if and only if the following three conditions hold:
\begin{eqnarray}
\label{twilled-1}\frac{1}{2}[\hat{\mu}_1,\hat{\mu}_1]_\B&=&0,\\
\label{twilled-2}[\hat{\mu}_1,\hat{\mu}_2]_\B&=&0,\\
\label{twilled-3}\frac{1}{2}[\hat{\mu}_2,\hat{\mu}_2]_\B&=&0.
\end{eqnarray}
\end{lem}
\begin{proof}
By Lemma \ref{important-lemma-2} and Lemma \ref{Zero-condition-1}, the proof is straightforward.
\end{proof}

\begin{pro}
  There is a one-to-one correspondence between matched pairs of Leibniz algebras and twilled Leibniz algebras.
\end{pro}

\begin{proof}
Let $(\g_1,\g_2;(\rho^L_{1},\rho^R_{1}),(\rho^L_{2},\rho^R_{2}))$ be a matched pair   of Leibniz algebras. By Proposition \ref{matched-pair-to-big-algebras}, we obtain that $(\g_1\oplus\g_2,[\cdot,\cdot]_{ \bowtie })$ is a Leibniz algebra. We denote this Leibniz algebra simply by $\g_1\bowtie\g_2$. Then $(\g_1\bowtie\g_2,\g_1,\g_2)$ is a twilled Leibniz algebra.

Conversely, if $(\huaG,\g_1,\g_2)$ is a twilled Leibniz algebra,  then $(\rho^L_{1},\rho^R_{1})$ is a representation of $\g_1$ on $\g_2$ and  $(\rho^L_{2},\rho^R_{2})$ is a representation of $\g_2$ on $\g_1$, where $\rho^L_{1},~\rho^R_{1},~ \rho^L_{2},~\rho^R_{2}$ are defined by
\begin{eqnarray*}
 \rho^L_{1}(x)u=[x,u]_2,\,\,\,\,\rho^R_{1}(x)u=[u,x]_2,\quad
\rho^L_{2}(u)x=[u,x]_1,\,\,\,\,\rho^R_{2}(u)x=[x,u]_1.
\end{eqnarray*}
  By Lemma \ref{lem:twillL},  $[\hat{\mu}_1,\hat{\mu}_2]_\B=0$, which is equivalent to  \eqref{matched-pair-1}-\eqref{matched-pair-6}. Thus,
 $(\g_1,\g_2;(\rho^L_{1},\rho^R_{1}),(\rho^L_{2},\rho^R_{2}))$ is a matched pair of Leibniz algebras.
\end{proof}

Let $(\huaG,[\cdot,\cdot]_\huaG)$ be a Leibniz algebra with a decomposition  into two subspaces, $\huaG=\g_1\oplus\g_2$, and $\Omega=\hat{\phi}_1+\hat{\mu}_1+\hat{\mu}_2+\hat{\phi}_2$ the Leibniz multiplication.
  Let $\hat{H}$ be the lift of a linear map $H:\g_2\lon\g_1$. Then $e^{[\cdot,\hat{H}]_\B}$ is an automorphism of the graded Lie algebra $(C^*(\huaG,\huaG),[\cdot,\cdot]_\B)$.

\begin{defi}
The transformation $\Omega^{H}:=e^{[\cdot,\hat{H}]_\B}\Omega$ is called a {\bf twisting} of $\Omega$ by $H$.
\end{defi}

\begin{lem}
$\Omega^{H}=e^{-\hat{H}}\circ \Omega\circ (e^{\hat{H}}\otimes e^{\hat{H}})$.
\end{lem}
\begin{proof}
For all $(x_1,v_1),~(x_2,v_2)\in\huaG$, we have
\begin{eqnarray*}
&&[\Omega,\hat{H}]_\B\big((x_1,v_1),(x_2,v_2)\big)=(\Omega\bar{\circ}\hat{H}-\hat{H}\bar{\circ}\Omega)\big((x_1,v_1),(x_2,v_2)\big)\\
                                                 &&\qquad\qquad=\Omega((H(v_1),0),(x_2,v_2))+\Omega((x_1,v_1),(H(v_2),0))-\hat{H}(\Omega((x_1,v_1),(x_2,v_2))).
\end{eqnarray*}
By $\hat{H}\circ\hat{H}=0$, we have
\begin{eqnarray*}
[[\Omega,\hat{H}]_\B,\hat{H}]_\B\big((x_1,v_1),(x_2,v_2)\big)&=&[\Omega,\hat{H}]_\B((H(v_1),0),(x_2,v_2))+[\Omega,\hat{H}]_\B((x_1,v_1),(H(v_2),0))\\
&&-\hat{H}\big([\Omega,\hat{H}]_\B((x_1,v_1),(x_2,v_2))\big)\\
&=&2\Omega((H(v_1),0),(H(v_2),0))-2\hat{H}\Omega((H(v_1),0),(x_2,v_2))\\
&&-2\hat{H}\Omega((x_1,v_1),(H(v_2),0)).
\end{eqnarray*}
Moreover, we have
\begin{eqnarray*}
{}[[[\Omega,\hat{H}]_\B,\hat{H}]_\B,\hat{H}]_\B\big((x_1,v_1),(x_2,v_2)\big)&=&-6\hat{H}\Omega((H(v_1),0),(H(v_2),0)),\\
{}\underbrace{[\cdots[[}_{i}\Omega,\hat{H}]_\B,\hat{H}]_\B,\cdots,\hat{H}]_\B\big((x_1,v_1),(x_2,v_2)\big)&=&0,\,\,\,\,\forall i\ge4,
\end{eqnarray*}
and
\begin{eqnarray}\label{eq:twisting-operator}
e^{[\cdot,\hat{H}]_\B}\Omega=\Omega+[\Omega,\hat{H}]_\B+\half[[\Omega,\hat{H}]_\B,\hat{H}]_\B+\frac{1}{6}[[[\Omega,\hat{H}]_\B,\hat{H}]_\B,\hat{H}]_\B.
\end{eqnarray}
Thus, we have
\begin{eqnarray*}
\Omega^{H}&=&\Omega-\hat{H}\circ\Omega+\Omega\circ(\hat{H}\otimes{\Id})+\Omega\circ({\Id}\otimes\hat{H})-\hat{H}\circ\Omega\circ({\Id}\otimes\hat{H})-\hat{H}\circ\Omega\circ(\hat{H}\otimes{\Id})\\
          \nonumber&&+\Omega\circ(\hat{H}\otimes\hat{H})-\hat{H}\circ\Omega\circ(\hat{H}\otimes\hat{H}).
\end{eqnarray*}
By $\hat{H}\circ\hat{H}=0$, we have
\begin{eqnarray*}
e^{-\hat{H}}\circ \Omega\circ (e^{\hat{H}}\otimes e^{\hat{H}})&=&({\Id}-\hat{H})\circ\Omega\circ (({\Id}+\hat{H})\otimes({\Id}+\hat{H}))\\
&=&\Omega+\Omega\circ({\Id}\otimes\hat{H})+\Omega\circ(\hat{H}\otimes{\Id})+\Omega\circ(\hat{H}\otimes\hat{H})\\
&&-\hat{H}\circ\Omega-\hat{H}\circ\Omega\circ({\Id}\otimes\hat{H})-\hat{H}\circ\Omega\circ(\hat{H}\otimes{\Id})-\hat{H}\circ\Omega\circ(\hat{H}\otimes\hat{H}).
\end{eqnarray*}
Thus, we obtain that $\Omega^{H}=e^{-\hat{H}}\circ \Omega\circ (e^{\hat{H}}\otimes e^{\hat{H}})$. The proof is finished.
\end{proof}

\begin{pro}
The twisting $\Omega^{H}$ is a Leibniz algebra structure on $\huaG$.
\end{pro}
\begin{proof}
By $\Omega^{H}=e^{-\hat{H}}\circ \Omega\circ (e^{\hat{H}}\otimes e^{\hat{H}})$, we have
\begin{eqnarray*}
[\Omega^{H},\Omega^{H}]_\B=2\Omega^{H}\bar{\circ}\Omega^{H}&=&2e^{-\hat{H}}\circ(\Omega\bar{\circ}\Omega)\circ(e^{\hat{H}}\otimes e^{\hat{H}}\otimes e^{\hat{H}})\\
&=&e^{-\hat{H}}\circ[\Omega,\Omega]_\B\circ(e^{\hat{H}}\otimes e^{\hat{H}}\otimes e^{\hat{H}})=0,
\end{eqnarray*}
which implies that $\Omega^{H}$ is a Leibniz algebra structure on $\huaG$ by Theorem \ref{leibniz-algebra-B}.
\end{proof}

\begin{cor}\label{twisting-isomorphism}
$
e^{\hat{H}}:(\huaG,\Omega^{H})\lon(\huaG,\Omega)
$
is an isomorphism between Leibniz algebras.
\end{cor}


Obviously,  $\Omega^{H}$ is also decomposed into the unique four substructures. The twisting operations are completely determined by the following result.

\begin{pro}\label{thm:twist}
Write $\Omega:=\hat{\phi}_1+\hat{\mu}_1+\hat{\mu}_2+\hat{\phi}_2$ and $\Omega^{H}:=\hat{\phi}_1^{H}+\hat{\mu}_1^{H}+\hat{\mu}_2^{H}+\hat{\phi}_2^{H}$. Then we have:
\begin{eqnarray}
\label{twisting-1}\hat{\phi}_1^{H}&=&\hat{\phi}_1,\\
\label{twisting-2}\hat{\mu}_1^{H}&=&\hat{\mu}_1+[\hat{\phi}_1,\hat{H}]_\B,\\
\label{twisting-3}\hat{\mu}_2^{H}&=&\hat{\mu}_2+[\hat{\mu}_1,\hat{H}]_\B+\half[[\hat{\phi}_1,\hat{H}]_\B,\hat{H}]_\B,\\
\label{twisting-4}\hat{\phi}_2^{H}&=&\hat{\phi}_2+[\hat{\mu}_2,\hat{H}]_\B+\half[[\hat{\mu}_1,\hat{H}]_\B,\hat{H}]_\B+\frac{1}{6}[[[\hat{\phi}_1,\hat{H}]_\B,\hat{H}]_\B,\hat{H}]_\B.
\end{eqnarray}
\end{pro}
\begin{proof}
By \eqref{eq:twisting-operator}, the first term of $\Omega^{H}$ is $\Omega$. By Lemma \ref{Zero-condition-2} and $||\hat{\phi}_2||=-1|2,~||\hat{H}||=-1|1$, the second term is
$$
[\hat{\phi}_1,\hat{H}]_\B+[\hat{\mu}_1,\hat{H}]_\B+[\hat{\mu}_2,\hat{H}]_\B.
$$
By Lemma \ref{important-lemma-2} and $||\hat{\phi}_1||=2|-1,~||\hat{\mu}_1||=1|0,~||\hat{\mu}_2||=0|1$, we have $$||[\hat{\phi}_1,\hat{H}]_\B||=1|0,\quad ||[\hat{\mu}_1,\hat{H}]_\B||=0|1,\quad||[\hat{\mu}_2,\hat{H}]_\B||=-1|2.$$
Therefore, $[[\hat{\mu}_2,\hat{H}]_\B,\hat{H}]_\B=0$ and the third term is
$$
\half([[\hat{\phi}_1,\hat{H}]_\B,\hat{H}]_\B+[[\hat{\mu}_1,\hat{H}]_\B,\hat{H}]_\B).
$$
Moreover, we have
$$||[[\hat{\phi}_1,\hat{H}]_\B,\hat{H}]_\B||=0|1,\quad||[[\hat{\mu}_1,\hat{H}]_\B,\hat{H}]_\B||=-1|2.$$
Thus, $[[[\hat{\mu}_1,\hat{H}]_\B,\hat{H}]_\B,\hat{H}]_\B=0$ and the final term is
$$
\frac{1}{6}[[[\hat{\phi}_1,\hat{H}]_\B,\hat{H}]_\B,\hat{H}]_\B.
$$
We have
$
||[[[\hat{\phi}_1,\hat{H}]_\B,\hat{H}]_\B,\hat{H}]_\B||=-1|2.
$
By Lemma \ref{Zero-condition-1}, the sum of all $-1|2$-terms is
$$
\hat{\phi}_2+[\hat{\mu}_2,\hat{H}]_\B+\half[[\hat{\mu}_1,\hat{H}]_\B,\hat{H}]_\B+\frac{1}{6}[[[\hat{\phi}_1,\hat{H}]_\B,\hat{H}]_\B,\hat{H}]_\B.
$$
Thus, we deduce that \eqref{twisting-4} holds. The sum of all $0|1$-terms is
$
\hat{\mu}_2+[\hat{\mu}_1,\hat{H}]_\B+\half[[\hat{\phi}_1,\hat{H}]_\B,\hat{H}]_\B.
$
Thus, we deduce  that \eqref{twisting-3} holds. The sum of all $1|0$-terms is
$
\hat{\mu}_1+[\hat{\phi}_1,\hat{H}]_\B.
$
Thus, we deduce that \eqref{twisting-2} holds. The sum of all $2|-1$-terms is
$
\hat{\phi}_1.
$
Thus, we deduce that \eqref{twisting-1} holds. The proof is finished.
\end{proof}

In the sequel, we consider a special case of the above twisting theory. Let $(V;\rho^L,\rho^R)$ be  a representation of a Leibniz algebra $(\g,[\cdot,\cdot]_\g)$. Consider the twilled Leibniz algebra  $(\g\ltimes_{\rho^L,\rho^R}V,\g,V)$.
Denote the Leibniz bracket $[\cdot,\cdot]_{\ltimes}$ by $\Omega$.
  Write $\Omega=\hat{\mu}_1+\hat{\mu}_2$. Then $\hat{\mu}_2=0$.

\begin{thm}\label{twisting-twilled}
With the above notations, let $H:V \longrightarrow\g$ be a linear map. The   twisting $((\g\oplus V, \Omega^{H}),\g,V)$ is a twilled Leibniz algebra if and only if $H$ is a  relative Rota-Baxter operator on the Leibniz algebra $(\g,[\cdot,\cdot]_\g)$ with respect to the representation $(V;\rho^L,\rho^R)$. Moreover, the Leibniz algebra structure on $V$ is given by \begin{eqnarray}\label{eq:mul2}
[u,v]_{H}:=\rho^L(H(u))v+\rho^R(H(v))u,\quad\forall u,v\in V.
\end{eqnarray}
\end{thm}
\begin{proof}
By Proposition \ref{thm:twist},
 the twisting have the form:
\begin{eqnarray}
\label{twilled-twisting-1}\hat{\mu}_1^{H}&=&\hat{\mu}_1,\\
\label{twilled-twisting-2}\hat{\mu}_2^{H}&=& [\hat{\mu}_1,\hat{H}]_\B,\\
\label{twilled-twisting-3}\hat{\phi}_2^{H}&=& \half[[\hat{\mu}_1,\hat{H}]_\B,\hat{H}]_\B.
\end{eqnarray}
Thus, the twisting $((\g\oplus V, \Omega^{H}),\g,V)$ is a twilled Leibniz algebra if and only if  $\hat{\phi}_2^{H}=0$, which implies that $H$ is a relative Rota-Baxter operator by Theorem \ref{twilled-DGLA}.

By Lemma \ref{lem:twillL}, we deduce that $\hat{\mu}_2^{H}$ is a Leibniz algebra multiplication on $V$. It is straightforward to deduce that the multiplication on $V$ is given by \eqref{eq:mul2}.
\end{proof}

\section{The classical Leibniz Yang-Baxter equation and triangular Leibniz bialgebras}\label{sec:R}

In this section, first we construct a Leibniz bialgebra using a symmetric relative Rota-Baxter operator. Then we define the classical Leibniz Yang-Baxter equation using the graded Lie algebra obtained in  Theorem \ref{twilled-DGLA}. Its solutions are called classical Leibniz $r$-matrices. Using the twisting theory given in Section \ref{sec:K}, we define a triangular Leibniz bialgebra successfully. 
Finally, we generalize a Semonov-Tian-Shansky's result in \cite{STS} about the relation between the operator form and the tensor form of a classical $r$-matrix to the context of Leibniz algebras.

Let $K:\g^*\lon\g$ be a relative Rota-Baxter operator on a Leibniz algebra $(\g,[\cdot,\cdot]_{\g})$ with respect to the representation $(\g^*;L^*,-L^*-R^*)$. Let $\Omega$ be the Leibniz bracket of the semidirect product Leibniz algebra $\g\ltimes_{L^*,-L^*-R^*}\g^*$. By Theorem \ref{twisting-twilled}, $(( \g\oplus\g^*,\Omega^K),\g,\g^*)$  is a twilled Leibniz algebra. Moreover, by Corollary \ref{twisting-isomorphism},
$e^{\hat{K}}:( \g\oplus\g^*,\Omega^K)\lon( \g\oplus\g^*,\Omega)$ is an isomorphism between  Leibniz algebras.

First by Theorem \ref{twisting-twilled}, we have

\begin{cor}\label{cor:algrep}
Let $K:\g^*\lon\g$ be a relative Rota-Baxter operator on $\g$ with respect to the representation $(\g^*;L^*,-L^*-R^*)$. Then $\g^*_K:=(\g^*,[\cdot,\cdot]_K)$ is a Leibniz algebra, where $[\cdot,\cdot]_K$ is given by
$$
[\xi,\eta]_K=L^*_{K\xi}\eta-L^*_{K\eta}\xi-R^*_{K\eta}\xi,\quad\forall\xi,\eta\in\g^*.
$$
\end{cor}

\begin{pro}
Let $K:\g^*\lon\g$ be a relative Rota-Baxter operator on  a Leibniz algebra $(\g,[\cdot,\cdot]_{\g})$ with respect to the representation $(\g^*;L^*,-L^*-R^*)$. Then $e^{\hat{K}}$ preserves the bilinear form $\omega$ given by \eqref{phase-space} if and only if $K^*$$=K.$ Here $K^*$ is the   dual map of $K$, i.e.
$\langle K\xi,\eta\rangle=\langle \xi,K^*\eta\rangle,$ for all $\xi,\eta\in\g^*.$
\end{pro}

\begin{proof}
By $\hat{K}\circ\hat{K}=0$, we have $e^{\hat{K}}={\Id}+\hat{K}$. For all $x,y\in\g,~\xi,\eta\in\g^*$, we have
\begin{eqnarray*}
\omega(e^{\hat{K}}(x+\xi),e^{\hat{K}}(y+\eta))&=&\omega(x+\xi+K(\xi),y+\eta+K(\eta))\\
                                              &=&\omega(x+\xi,y+\eta)+\omega(x+\xi,K(\eta))+\omega(K(\xi),y+\eta)+\omega(K(\xi),K(\eta))\\
                                              &=&\omega(x+\xi,y+\eta)+\omega(\xi,K(\eta))+\omega(K(\xi),\eta)\\
                                              &=&\omega(x+\xi,y+\eta)+\langle \xi,K(\eta)\rangle-\langle \eta,K(\xi)\rangle\\
                                              &=&\omega(x+\xi,y+\eta)+\langle (K^*-K)\xi,\eta\rangle.
\end{eqnarray*}
Thus, $\omega(e^{\hat{K}}(x+\xi),e^{\hat{K}}(y+\eta))=\omega(x+\xi,y+\eta)$ if and only if $K^*=K.$
\end{proof}

\begin{pro}\label{isomorphism-invariant-twilled}
Let $K:\g^*\lon\g$ be a relative Rota-Baxter operator on  a Leibniz algebra $(\g,[\cdot,\cdot]_{\g})$ with respect to the representation $(\g^*;L^*,-L^*-R^*)$ and $K^*=K$.
Then $(\g\oplus\g^*,\Omega^K)$  is a quadratic Leibniz algebra with the invariant bilinear form $\omega$ given by \eqref{phase-space} and $e^{\hat{K}}$ is an isomorphism from the quadratic Leibniz algebra $(\g\oplus\g^*,\Omega^K)$ to $(\g\oplus\g^*,\Omega)$.
 \end{pro}

\begin{proof}
Since $e^{\hat{K}}$ is a Leibniz algebra isomorphism and preserves the bilinear form $\omega$, for all $X,Y,Z\in\g\oplus\g^*$, we have
\begin{eqnarray*}
  \omega(X,\Omega^K(Y,Z))&=& \omega(X,e^{-\hat{K}}\Omega(e^{\hat{K}}Y,e^{\hat{K}}Z))=\omega(e^{\hat{K}}X,\Omega(e^{\hat{K}}Y,e^{\hat{K}}Z))\\
  &=&\omega(\Omega(e^{\hat{K}}X,e^{\hat{K}}Z)+\Omega(e^{\hat{K}}Z,e^{\hat{K}}X),e^{\hat{K}}Y)\\
  &=&\omega(\Omega^K(X,Z)+\Omega^K(Z,X),Y),
\end{eqnarray*}
which implies that $(\g\oplus\g^*,\Omega^K)$  is a quadratic   Leibniz algebra. It is obvious that $e^{\hat{K}}$ is an isomorphism from the quadratic Leibniz algebra $(\g\oplus\g^*,\Omega^K)$ to $(\g\oplus\g^*,\Omega)$.
\end{proof}

By Corollary \ref{cor:algrep}, Proposition \ref{isomorphism-invariant-twilled} and Theorem \ref{thm:equivalent}, we obtain
\begin{thm}\label{cor:bialg}
 Let $K:\g^*\lon\g$ be a relative Rota-Baxter operator on  a Leibniz algebra $(\g,[\cdot,\cdot]_{\g})$ with respect to the representation $(\g^*;L^*,-L^*-R^*)$ and $K^*=K$. Then $(\g,\g^*_K)$ is a Leibniz bialgebra, where the Leibniz algebra $\g_K^*$ is given in Corollary \ref{cor:algrep}.
\end{thm}

By Theorem \ref{twilled-DGLA}, we have
\begin{cor}\label{dual-GLA}
Let $(\g,[\cdot,\cdot]_\g)$ be a Leibniz algebra. Then   $(C^*(\g^*,\g),\{\cdot,\cdot\})$ is a graded Lie algebra, where $\{\cdot,\cdot\}$ is given by
\begin{eqnarray*}
&&\nonumber\{g_1,g_2\}(\xi_1,\xi_2,\cdots,\xi_{m+n})\\
&=&\sum_{k=1}^{m}\sum_{\sigma\in\mathbb S_{(k-1,n)}}(-1)^{(k-1)n+1}(-1)^{\sigma}g_1(\xi_{\sigma(1)},\cdots,\xi_{\sigma(k-1)},L^*_{g_2(\xi_{\sigma(k)},\cdots,\xi_{\sigma(k+n-1)})}\xi_{k+n},\xi_{k+n+1},\cdots,\xi_{m+n})\\
&&+\sum_{k=2}^{m+1}\sum_{\sigma\in\mathbb S_{(k-2,n,1)}\atop \sigma(k+n-2)=k+n-1}(-1)^{kn+1}(-1)^{\sigma}
g_1(\xi_{\sigma(1)},\cdots,\xi_{\sigma(k-2)},L^*_{g_2(\xi_{\sigma(k-1)},\cdots,\xi_{\sigma(k+n-2)})}\xi_{\sigma(k+n-1)},\xi_{k+n},\cdots,\xi_{m+n})\\
&&+\sum_{k=2}^{m+1}\sum_{\sigma\in\mathbb S_{(k-2,n,1)}\atop \sigma(k+n-2)=k+n-1}(-1)^{kn+1}(-1)^{\sigma}
g_1(\xi_{\sigma(1)},\cdots,\xi_{\sigma(k-2)},R^*_{g_2(\xi_{\sigma(k-1)},\cdots,\xi_{\sigma(k+n-2)})}\xi_{\sigma(k+n-1)},\xi_{k+n},\cdots,\xi_{m+n})\\
&&+\sum_{k=1}^{m}\sum_{\sigma\in\mathbb S_{(k-1,n-1)}}(-1)^{(k-1)n}(-1)^{\sigma}[g_2(\xi_{\sigma(k)},\cdots,\xi_{\sigma(k+n-2)},\xi_{k+n-1}),g_1(\xi_{\sigma(1)},\cdots,\xi_{\sigma(k-1)},\xi_{k+n},\cdots,\xi_{m+n})]_{\g}\\
&&+\sum_{\sigma\in\mathbb S_{(m,n-1)}}(-1)^{mn+1}(-1)^{\sigma}[g_1(\xi_{\sigma(1)},\cdots,\xi_{\sigma(m)}),g_2(\xi_{\sigma(m+1)},\cdots,\xi_{\sigma(m+n-1)},\xi_{m+n})]_{\g}\\
&&+\sum_{k=1}^{n}\sum_{\sigma\in\mathbb S_{(k-1,m)}}(-1)^{m(k+n-1)}(-1)^{\sigma}g_2(\xi_{\sigma(1)},\cdots,\xi_{\sigma(k-1)},L^*_{g_1(\xi_{\sigma(k)},\cdots,\xi_{\sigma(k+m-1)})}\xi_{k+m},\xi_{k+m+1},\cdots,\xi_{m+n})\\
&&+\sum_{k=1}^{n}\sum_{\sigma\in\mathbb S_{(k-1,m,1)}\atop\sigma(k+m-1)=k+m}(-1)^{m(k+n-1)}(-1)^{\sigma}
g_2(\xi_{\sigma(1)},\cdots,\xi_{\sigma(k-1)},L^*_{g_1(\xi_{\sigma(k)},\cdots,\xi_{\sigma(k-1+m)})}\xi_{\sigma(k+m)},\xi_{k+m+1},\cdots,\xi_{m+n})\\
&&+\sum_{k=1}^{n}\sum_{\sigma\in\mathbb S_{(k-1,m,1)}\atop\sigma(k+m-1)=k+m}(-1)^{m(k+n-1)}(-1)^{\sigma}
g_2(\xi_{\sigma(1)},\cdots,\xi_{\sigma(k-1)},R^*_{g_1(\xi_{\sigma(k)},\cdots,\xi_{\sigma(k-1+m)})}\xi_{\sigma(k+m)},\xi_{k+m+1},\cdots,\xi_{m+n}).
\end{eqnarray*}
\end{cor}

 In the sequel, to define the classical Leibniz Yang-Baxter equation, we transfer the above graded Lie algebra structure to the tensor space.

For $k\ge1$, we define $\Psi:\otimes^{k+1}\g\longrightarrow \Hom(\otimes^k\g^*,\g)$ by
\begin{equation}\label{eq:defipsi}
 \langle\Psi(P)(\xi_1,\cdots,\xi_k),\xi_{k+1}\rangle=\langle P,~\xi_1\otimes\cdots\otimes\xi_k\otimes\xi_{k+1}\rangle,\quad \forall P\in\otimes^{k+1}\g,~ \xi_1,\cdots, \xi_{k+1}\in\g^*,
\end{equation}
and $\Upsilon:\Hom(\otimes^k\g^*,\g)\longrightarrow \otimes^{k+1}\g$ by
\begin{equation}\label{eq:defiUpsilon}
 \langle\Upsilon(f),\xi_1\otimes\cdots\otimes\xi_k\otimes\xi_{k+1}\rangle=\langle f(\xi_1,\cdots,\xi_k),\xi_{k+1}\rangle,\quad \forall f\in\Hom(\otimes^k\g^*,\g),~ \xi_1,\cdots, \xi_{k+1}\in\g^*.
\end{equation}
Obviously we have $\Psi\circ\Upsilon={\Id},~~\Upsilon\circ\Psi={\Id}.$

\begin{thm}
Let $(\g,[\cdot,\cdot]_\g)$ be a Leibniz algebra. Then, there is a graded Lie  bracket $[[\cdot,\cdot]]$ on the graded space $\oplus_{k\ge2}(\otimes^{k}\g)$ given by
$$
[[ P,Q]]:=\Upsilon\{\Psi(P),\Psi(Q)\},\,\,\,\,\forall P\in\otimes^{m+1}\g,Q\in\otimes^{n+1}\g.
$$
\end{thm}

\begin{proof}
By $\Psi\circ\Upsilon={\Id},~~\Upsilon\circ\Psi={\Id},$ we transfer the graded Lie algebra structure on $C^*(\g^*,\g)$ to that on the graded  space $\oplus_{k\ge2}(\otimes^{k}\g)$. The proof is finished.
\end{proof}
The general formula of $[[P,Q]]$ is very sophisticated. But for $P=x\otimes y$ and $Q=z\otimes w$, there is an explicit expression, which is enough for our application.
\begin{lem}\label{gla-r-matrix}
For $x\otimes y,~z\otimes w\in\g\otimes\g$, we have
\begin{eqnarray}\label{2-tensor}
\nonumber[[x\otimes y,z\otimes w]]&=&z\otimes[w,x]_\g\otimes y-[w,x]_\g\otimes z\otimes y-[x,w]_\g\otimes z\otimes y+z\otimes x\otimes[w,y]_\g\\
                       &&+x\otimes z\otimes[y,w]_\g+x\otimes[y,z]_\g\otimes w-[y,z]_\g\otimes x\otimes w-[z,y]_\g\otimes x\otimes w.
\end{eqnarray}
\end{lem}

\begin{proof}
For all $\xi\in\g^*$, we have $\Psi(x\otimes y)(\xi)=\langle x,\xi\rangle y$. By Corollary \ref{dual-GLA}, for all $\xi_1,\xi_2\in\g^*$, we have
\begin{eqnarray*}
\{\Psi(x\otimes y),\Psi(z\otimes w)\}(\xi_1,\xi_2)
&=&-\Psi(x\otimes y)(L^*_{\Psi(z\otimes w)\xi_1}\xi_2)+\Psi(x\otimes y)(L^*_{\Psi(z\otimes w)\xi_2}\xi_1)\\
&&+\Psi(x\otimes y)(R^*_{\Psi(z\otimes w)\xi_2}\xi_1)+[\Psi(z\otimes w)\xi_1,\Psi(x\otimes y)\xi_2]_\g\\
                                      &&+[\Psi(x\otimes y)\xi_1,\Psi(z\otimes w)\xi_2]_\g-\Psi(z\otimes w)(L^*_{\Psi(x\otimes y)\xi_1}\xi_2)\\
                                      &&+\Psi(z\otimes w)(L^*_{\Psi(x\otimes y)\xi_2}\xi_1)+\Psi(z\otimes w)(R^*_{\Psi(x\otimes y)\xi_2}\xi_1).
\end{eqnarray*}
Thus, for all $\xi_1,\xi_2,\xi_3\in\g^*$, we have
\begin{eqnarray*}
&&\langle [[x\otimes y,z\otimes w]],\xi_1\otimes\xi_2\otimes\xi_3\rangle=\langle \{\Psi(x\otimes y),\Psi(z\otimes w)\}(\xi_1,\xi_2),\xi_3\rangle\\
                                                                    &=&-\langle\Psi(x\otimes y)(L^*_{\Psi(z\otimes w)\xi_1}\xi_2),\xi_3\rangle+\langle\Psi(x\otimes y)(L^*_{\Psi(z\otimes w)\xi_2}\xi_1),\xi_3\rangle
                                                                    +\langle\Psi(x\otimes y)(R^*_{\Psi(z\otimes w)\xi_2}\xi_1),\xi_3\rangle\\
                                                                    && +\langle[\Psi(z\otimes w)\xi_1,\Psi(x\otimes y)\xi_2]_\g,\xi_3\rangle+\langle[\Psi(x\otimes y)\xi_1,\Psi(z\otimes w)\xi_2]_\g,\xi_3\rangle-\langle\Psi(z\otimes w)(L^*_{\Psi(x\otimes y)\xi_1}\xi_2),\xi_3\rangle\\
                                                                    && +\langle\Psi(z\otimes w)(L^*_{\Psi(x\otimes y)\xi_2}\xi_1),\xi_3\rangle+\langle\Psi(z\otimes w)(R^*_{\Psi(x\otimes y)\xi_2}\xi_1),\xi_3\rangle\\
                                                                    &=&-\langle z,\xi_1\rangle\langle x,L^*_w\xi_2\rangle\langle y,\xi_3\rangle+\langle z,\xi_2\rangle\langle x,L^*_w\xi_1\rangle\langle y,\xi_3\rangle+\langle z,\xi_2\rangle\langle x,R^*_w\xi_1\rangle\langle y,\xi_3\rangle\\
                                                                    && +\langle z,\xi_1\rangle\langle x,\xi_2\rangle\langle [w,y]_\g,\xi_3\rangle+\langle x,\xi_1\rangle\langle z,\xi_2\rangle\langle [y,w]_\g,\xi_3\rangle-\langle x,\xi_1\rangle\langle z,L^*_y\xi_2\rangle\langle w,\xi_3\rangle\\
                                                                    && +\langle x,\xi_2\rangle\langle z,L^*_y\xi_1\rangle\langle w,\xi_3\rangle+\langle x,\xi_2\rangle\langle z,R^*_y\xi_1\rangle\langle w,\xi_3\rangle\\
                                                                    &=&\langle z,\xi_1\rangle\langle[w,x]_\g,\xi_2\rangle\langle y,\xi_3\rangle-\langle [w,x]_\g,\xi_1\rangle\langle z,\xi_2\rangle\langle y,\xi_3\rangle-\langle[x,w]_\g,\xi_1\rangle\langle z,\xi_2\rangle\langle y,\xi_3\rangle\\
                                                                    &&+\langle z,\xi_1\rangle\langle x,\xi_2\rangle\langle [w,y]_\g,\xi_3\rangle +\langle x,\xi_1\rangle\langle z,\xi_2\rangle\langle [y,w]_\g,\xi_3\rangle+\langle x,\xi_1\rangle\langle [y,z]_\g,\xi_2\rangle\langle w,\xi_3\rangle\\
                                                                    &&-\langle[y,z]_\g,\xi_1\rangle\langle x,\xi_2\rangle\langle w,\xi_3\rangle-\langle[z,y]_\g,\xi_1\rangle\langle x,\xi_2\rangle\langle w,\xi_3\rangle,
\end{eqnarray*}
which implies that \eqref{2-tensor} holds.
\end{proof}

Moreover, we can obtain the tensor form of a relative Rota-Baxter operator on $\g$ with respect to the representation $(\g^*;L^*,-L^*-R^*)$.

\begin{pro}\label{o-operator-tensor-form}
Let $K:\g^*\lon\g$ be a linear map.
 \begin{itemize}
   \item[\rm(i)] $K$ is a relative Rota-Baxter operator on $\g$ with respect to the representation $(\g^*;L^*,-L^*-R^*)$ if and only if the tensor form $\bar{K}=\Upsilon(K)\in\g\otimes\g$ satisfies
\begin{eqnarray*}
\nonumber[[\bar{K},\bar{K}]] = 0.
\end{eqnarray*}
 \item[\rm(ii)] $K=K^*$ if and only if   $\bar{K}=\tau_{12}(\bar{K})$, that is, $\bar{K}\in\Sym^2(\g)$.
 \end{itemize}
\end{pro}

\begin{proof}
Since $\Psi$ is graded Lie algebra isomorphism from $(\oplus_{k\ge2}(\otimes^{k}\g),[[\cdot,\cdot]])$ to $(C^*(\g^*,\g),\{\cdot,\cdot\})$, we deduce that $\{K,K\}=0$ if and only if $[[\bar{K},\bar{K}]]=0.$ The other conclusion is obvious.
\end{proof}

 \emptycomment{
\begin{proof}
Let $\xi_1,\xi_2\in\g^*$, we have
\begin{eqnarray*}
\langle T\xi_1,\xi_2\rangle=\langle \xi_1,T^*\xi_2\rangle=\langle T^*\xi_2,\xi_1\rangle=\langle \bar{T^*},\xi_2\otimes\xi_1\rangle=\langle \tau_{12}(\bar{T^*}),\xi_1\otimes\xi_2\rangle.
\end{eqnarray*}
On the other hand, we have
\begin{eqnarray*}
\langle T\xi_1,\xi_2\rangle=\langle \bar{T},\xi_1\otimes\xi_2\rangle.
\end{eqnarray*}
Thus, we have $\bar{T}=\tau_{12}(\bar{T^*})$. We deduce that $T=T^*$ if and only if the tensor form $\bar{T}=\tau_{12}(\bar{T})$. The proof is finished.
\end{proof}
}

\begin{defi}
Let $(\g,[\cdot,\cdot]_\g)$ be a Leibniz algebra and $r\in\Sym^2(\g)$. Then equation
\begin{equation}\label{Leibniz Yang-Baxter}
~[[r,r]]=0
\end{equation} is
called the {\bf classical Leibniz Yang-Baxter equation} in $\g$ and $r$ is called a {\bf classical Leibniz $r$-matrix}.
\end{defi}

\begin{ex}\label{example-8}{\rm
Consider the $2$-dimensional Leibniz algebra $(\g,[\cdot,\cdot])$  defined with respect to a basis $\{e_1,e_2\}$  by
\begin{eqnarray*}
[e_1,e_1]=0,\quad [e_1,e_2]=0,\quad [e_2,e_1]=e_1,\quad [e_2,e_2]=e_1.
\end{eqnarray*}
 By Lemma \ref{gla-r-matrix}, we have
\begin{eqnarray*}
&&[[e_1\otimes e_1,e_1\otimes e_1]]=0,\\
&&[[e_1\otimes e_1,e_1\otimes e_2]]=e_1\otimes e_1\otimes e_1,\\
&&[[e_1\otimes e_1,e_2\otimes e_1]]=-e_1\otimes e_1\otimes e_1,\\
&&[[e_1\otimes e_1,e_2\otimes e_2]]=2e_2\otimes e_1\otimes e_1-e_1\otimes e_2\otimes e_1-e_1\otimes e_1\otimes e_2,\\
&&[[e_1\otimes e_2,e_1\otimes e_2]]=2e_1\otimes e_1\otimes e_1,\\
&&[[e_1\otimes e_2,e_2\otimes e_1]]=e_1\otimes e_2\otimes e_1-e_1\otimes e_1\otimes e_1,\\
&&[[e_1\otimes e_2,e_2\otimes e_2]]=e_2\otimes e_1\otimes e_2-e_1\otimes e_2\otimes e_2+e_2\otimes e_1\otimes e_1+e_1\otimes e_2\otimes e_1-e_1\otimes e_1\otimes e_2,\\
&&[[e_2\otimes e_1,e_2\otimes e_1]]=-2e_1\otimes e_2\otimes e_1,\\
&&[[e_2\otimes e_1,e_2\otimes e_2]]=e_2\otimes e_1\otimes e_1-2e_1\otimes e_2\otimes e_1+e_2\otimes e_2\otimes e_1-e_1\otimes e_2\otimes e_2,\\
&&[[e_2\otimes e_2,e_2\otimes e_2]]=2e_2\otimes e_1\otimes e_2-4e_1\otimes e_2\otimes e_2+2e_2\otimes e_2\otimes e_1.
\end{eqnarray*}
Therefor, for all $r=ae_1\otimes e_1+be_1\otimes e_2+be_2\otimes e_1+ce_2\otimes e_2\in\Sym^2(\g)$,   we have
\begin{eqnarray*}
&&[[ae_1\otimes e_1+be_1\otimes e_2+be_2\otimes e_1+ce_2\otimes e_2,ae_1\otimes e_1+be_1\otimes e_2+be_2\otimes e_1+ce_2\otimes e_2]]\\
&=&ab[[e_1\otimes e_1,e_1\otimes e_2]]+ab[[e_1\otimes e_1,e_2\otimes e_1]]+ac[[e_1\otimes e_1,e_2\otimes e_2]]+ba[[e_1\otimes e_2,e_1\otimes e_1]]\\
&&+b^2[[e_1\otimes e_2,e_1\otimes e_2]]+b^2[[e_1\otimes e_2,e_2\otimes e_1]]+bc[[e_1\otimes e_2,e_2\otimes e_2]]+ba[[e_2\otimes e_1,e_1\otimes e_1]]\\
&&+b^2[[e_2\otimes e_1,e_1\otimes e_2]]+b^2[[e_2\otimes e_1,e_2\otimes e_1]]+bc[[e_2\otimes e_1,e_2\otimes e_2]]+ca[[e_2\otimes e_2,e_1\otimes e_1]]\\
&&+cb[[e_2\otimes e_2,e_1\otimes e_2]]+cb[[e_2\otimes e_2,e_2\otimes e_1]]+c^2[[e_2\otimes e_2,e_2\otimes e_2]]\\
&=&2ab[[e_1\otimes e_1,e_1\otimes e_2]]+2ab[[e_1\otimes e_1,e_2\otimes e_1]]+2ac[[e_1\otimes e_1,e_2\otimes e_2]]\\
&&+b^2[[e_1\otimes e_2,e_1\otimes e_2]]+2b^2 [[e_1\otimes e_2,e_2\otimes e_1]]+2bc [[e_1\otimes e_2,e_2\otimes e_2]]\\
&&+b^2[[e_2\otimes e_1,e_2\otimes e_1]]+2bc[[e_2\otimes e_1,e_2\otimes e_2]]+c^2[[e_2\otimes e_2,e_2\otimes e_2]]\\
&=&2abe_1\otimes e_1\otimes e_1-2abe_1\otimes e_1\otimes e_1+2ac\Big(2e_2\otimes e_1\otimes e_1-e_1\otimes e_2\otimes e_1-e_1\otimes e_1\otimes e_2\Big)\\
&&+2b^2e_1\otimes e_1\otimes e_1+2b^2\Big(e_1\otimes e_2\otimes e_1-e_1\otimes e_1\otimes e_1\Big)\\
&&+2bc\Big(e_2\otimes e_1\otimes e_2-e_1\otimes e_2\otimes e_2+e_2\otimes e_1\otimes e_1+e_1\otimes e_2\otimes e_1-e_1\otimes e_1\otimes e_2\Big)-2b^2e_1\otimes e_2\otimes e_1\\
&&+2bc\Big(e_2\otimes e_1\otimes e_1-2e_1\otimes e_2\otimes e_1+e_2\otimes e_2\otimes e_1-e_1\otimes e_2\otimes e_2\Big)\\
&&+c^2\Big(2e_2\otimes e_1\otimes e_2-4e_1\otimes e_2\otimes e_2+2e_2\otimes e_2\otimes e_1\Big),\\
&=&4c(a+b)e_2\otimes e_1\otimes e_1-2c(a+b)e_1\otimes e_2\otimes e_1-2c(a+b)e_1\otimes e_1\otimes e_2\\&&+2c(b+c)e_2\otimes e_1\otimes e_2-4c(b+c)e_1\otimes e_2\otimes e_2+2c(b+c)e_2\otimes e_2\otimes e_1.
\end{eqnarray*}
\begin{itemize}
     \item[\rm(i)] If $c=0$, then  any $r=ae_1\otimes e_1+b(e_1\otimes e_2+e_2\otimes e_1)$ is a classical Leibniz $r$-matrix;
     \item[\rm(ii)] If $c\not=0$, then $[[r,r]]=0$ if and only if $a=c=-b$. Thus, any $$r=c\Big(e_1\otimes e_1-e_1\otimes e_2-e_2\otimes e_1+e_2\otimes e_2\Big)$$ is a classical Leibniz $r$-matrix.
   \end{itemize}
   }
\end{ex}

\begin{rmk}
For $r=ae_1\otimes e_1+be_1\otimes e_2+be_2\otimes e_1+ce_2\otimes e_2$, we have
$
r^{\sharp}(e_1^*,e_2^*)=(e_1,e_2)\left(\begin{array}{cc}a&b\\
                                                        b&c\end{array}\right).
$
The above classical Leibniz $r$-matrices actually correspond to   symmetric relative Rota-Baxter operators given in Example \ref{example-6}.
\end{rmk}

\begin{cor}
Let $(\g,[\cdot,\cdot]_\g)$ be a Leibniz algebra and $r\in\Sym^2(\g)$ a solution of the classical Leibniz Yang-Baxter equation in $\g$. Then $(\g,\g^*_{r^\sharp})$ is a Leibniz bialgebra, where $r^\sharp=\Psi(r):\g^*\lon\g$ is defined by $\langle r^\sharp(\xi),\eta\rangle=\langle r,\xi\otimes \eta\rangle$ for all $\xi,\eta\in\g^*$.
\end{cor}

\begin{proof}
By $r\in\Sym^2(\g)$ and $[[r,r]]=0$, we deduce that $r^\sharp:\g^*\lon\g$ is a relative Rota-Baxter operator on $\g$ with respect to the representation $(\g^*;L^*,-L^*-R^*)$ and $(r^\sharp)^*=r^\sharp$. By Theorem \ref{cor:bialg}, we obtain that $(\g,\g^*_{r^\sharp})$ is a Leibniz bialgebra.
\end{proof}

\begin{defi}
Let $(\g,[\cdot,\cdot]_\g)$ be a Leibniz algebra and $r\in\Sym^2(\g)$ a solution of the classical Leibniz Yang-Baxter equation in $\g$. We call the Leibniz bialgebra $(\g,\g^*_{r^\sharp})$ the {\bf triangular Leibniz bialgebra} associated to the classical Leibniz $r$-matrix $r$.
\end{defi}

\begin{rmk}
   In Section \ref{sec:B}, we define a Leibniz bialgebra, which is equivalent to a Manin triple of Leibniz algebras. Note that there is no cohomology theory can be used in the theory of Leibniz bialgebras. Thus, there is not an obvious way to define a ``coboundary Leibniz bialgebra''. Nevertheless, using the twisting method in the theory of twilled Leibniz algebras, we define triangular Leibniz bialgebras successfully.
\end{rmk}

In the Lie algebra context, we know that the dual description of a classical $r$-matrix is a symplectic structure on a Lie algebra. Now we  investigate the dual description of a classical Leibniz $r$-matrix. A symmetric 2-form $\huaB\in\Sym^2(\g^*)$ on a Leibniz algebra $(\g,[\cdot,\cdot]_\g)$ induces a linear map $\huaB^{\natural}:\g\lon\g^*$ by
\begin{eqnarray*}
\langle\huaB^{\natural}(x),y\rangle:=\huaB(x,y),\,\,\,\,\forall x,y\in\g.
\end{eqnarray*}
$\huaB$ is said to be nondegenerate if $\huaB^\natural:\g\lon\g^* $ is an isomorphism.  Similarly, $r\in\Sym^2(\g)$ is said to be nondegenerate if $r^\sharp:\g^*\lon\g$ is an isomorphism.

\begin{pro}\label{Hessian-structure-tensor-form}
 $r\in\Sym^2(\g)$ is a   nondegenerate solution of the classical Leibniz Yang-Baxter equation in a Leibniz algebra $\g$ if and only if the symmetric nondegenerate bilinear form $\huaB$ on
$\g$ defined by
\begin{eqnarray}
\huaB(x,y):=\langle (r^\sharp)^{-1}(x),y\rangle,\quad\forall x,y\in\g,
\end{eqnarray}
 satisfies the following ``closed'' condition
 \begin{eqnarray}\label{Hessian-structure}
\huaB(z,[x,y]_\g)=-\huaB(y,[x,z]_\g)+\huaB(x,[y,z]_\g)+\huaB(x,[z,y]_\g).
\end{eqnarray}
\end{pro}

\begin{proof}
Let $r\in\Sym^2(\g)$ be   nondegenerate. It is obvious that $\huaB$ is symmetric and nondegenerate.

Since $r^\sharp:\g^*\lon\g$ is an invertible linear map, for all $x,y,z\in\g$, there are $\xi_1,\xi_2,\xi_3\in\g^*$ such that $r^\sharp(\xi_1)=x,~r^\sharp(\xi_2)=y$ and $r^\sharp(\xi_3)=z$. Since $r^\sharp$ is a relative Rota-Baxter operator on $\g$ with respect to the representation $(\g^*;L^*,-L^*-R^*)$ and $(r^\sharp)^*=r^\sharp$, we have
\begin{eqnarray*}
\huaB(z,[x,y]_\g)&=&\langle (r^\sharp)^{-1}(z),[x,y]_\g\rangle=\langle \xi_3,[r^\sharp(\xi_1),r^\sharp(\xi_2)]_\g\rangle\\
                                                   &=&\langle \xi_3,r^\sharp(L^*_{r^\sharp(\xi_1)}\xi_2)\rangle-\langle \xi_3,r^\sharp(L^*_{r^\sharp(\xi_2)}\xi_1)\rangle-\langle \xi_3,r^\sharp(R^*_{r^\sharp(\xi_2)}\xi_1)\rangle\\
                                                   &=&\langle r^\sharp(\xi_3),L^*_{r^\sharp(\xi_1)}\xi_2\rangle-\langle r^\sharp(\xi_3),L^*_{r^\sharp(\xi_2)}\xi_1\rangle-\langle r^\sharp(\xi_3),R^*_{r^\sharp(\xi_2)}\xi_1\rangle\\
                                                   &=&-\langle [r^\sharp(\xi_1),r^\sharp(\xi_3)]_\g,\xi_2\rangle+\langle [r^\sharp(\xi_2),r^\sharp(\xi_3)]_\g,\xi_1\rangle+\langle [r^\sharp(\xi_3),r^\sharp(\xi_2)]_\g,\xi_1\rangle\\
                                                   &=&-\huaB([x,z]_\g,y)+\huaB([z,y]_\g,x)+\huaB([y,z]_\g,x).
\end{eqnarray*}
Thus, $\huaB$ satisfies \eqref{Hessian-structure}.
\end{proof}

At the end of this section, we generalize a Semonov-Tian-Shansky's result in \cite{STS} to the context of Leibniz algebras.

\begin{lem}\label{iso-rep}
Let $(\g,[\cdot,\cdot]_\g,\omega)$ be a quadratic Leibniz algebra. Then $\omega^{\natural}:\g\lon\g^*$ is an isomorphism from the regular representation $(\g;L,R)$ to its dual representation $(\g^*;L^*,-L^*-R^*)$.
\end{lem}

\begin{proof}
For all $x,y,z\in\g$, by \eqref{superfluous-condition} we have
\begin{eqnarray*}
\langle\omega^{\natural}(L_xy)-L^*_x\omega^{\natural}(y),z\rangle=\omega([x,y]_\g,z)+\langle \omega^{\natural}(y),L_xz\rangle=\omega([x,y]_\g,z)+\omega(y,[x,z]_\g)=0.
\end{eqnarray*}
Thus, we have $\omega^{\natural}\circ L_x=L^*_x\circ\omega^{\natural}.$ By \eqref{Invariant-bilinear-forms}, we have
\begin{eqnarray*}
\langle\omega^{\natural}(R_xy)-(-L^*_x-R^*_x)\omega^{\natural}(y),z\rangle&=&\omega([y,x]_\g,z)-\langle\omega^{\natural}(y),(L_x+R_x)z\rangle\\
                                                                           &=&\omega([y,x]_\g,z)-\omega(y,[x,z]_\g+[z,x]_\g)\\
                                                                           &=&0.
\end{eqnarray*}
Thus, we have $\omega^{\natural}\circ R_x=(-L^*_x-R^*_x)\circ\omega^{\natural}.$ Therefore,   $\omega^{\natural}$ is an isomorphism between representations. The proof is finished.
\end{proof}

\begin{thm}
Let $(\g,[\cdot,\cdot]_\g,\omega)$ be a quadratic Leibniz algebra and $K:\g^*\lon\g$ a linear map. Then $K$ is a relative Rota-Baxter operator on   $(\g,[\cdot,\cdot]_{\g})$ with respect to the representation $(\g^*;L^*,-L^*-R^*)$ if and only if $K\circ\omega^{\natural}$ is a Rota-Baxter operator on $(\g,[\cdot,\cdot]_{\g})$.
\end{thm}
\begin{proof}
For all $x,y\in\g$, by Lemma \ref{iso-rep}, we have
\begin{eqnarray*}
(K\circ\omega^{\natural})([K\omega^{\natural}(x),y]_\g+[x,K\omega^{\natural}(y)]_\g)&=&K\Big(\omega^{\natural}(L_{K\omega^{\natural}(x)}y)+\omega^{\natural}R_{K\omega^{\natural}(y)}x\Big)\\
                                                                                    &=&K\Big(L^*_{K\omega^{\natural}(x)}\omega^{\natural}(y)-L^*_{K\omega^{\natural}(y)}\omega^{\natural}(x)
                                                                                    -R^*_{K\omega^{\natural}(y)}\omega^{\natural}(x)\Big),
\end{eqnarray*}
which implies the conclusion.
\end{proof}

\begin{cor}
Let $(\g,[\cdot,\cdot]_\g,\omega)$ be a quadratic Leibniz algebra. Then $r\in\Sym^2(\g)$ is a solution of the classical Leibniz Yang-Baxter equation in $\g$ if and only if    $r^{\sharp}\circ\omega^{\natural}$ is a Rota-Baxter operator on   $(\g,[\cdot,\cdot]_{\g})$, that is,
$$
[r^{\sharp}\circ\omega^{\natural}(x),r^{\sharp}\circ\omega^{\natural}(y)]_\g=(r^{\sharp}\circ\omega^{\natural})[r^{\sharp}\circ\omega^{\natural}(x),y]_\g+(r^{\sharp}\circ\omega^{\natural})[x,r^{\sharp}\circ\omega^{\natural}(y)]_\g.
$$
\end{cor}

\begin{rmk}
  In \cite{Rmatrix}, the authors defined $R_{\pm}$-matrix for Leibniz algebras as a direct generalization of Semonov-Tian-Shansky's approach in \cite{STS}, without any bialgebra theory for Leibniz algebras. It is straightforward to see that their $R_+$-matrices in a  Leibniz algebra are simply Rota-Baxter operator~on the Leibniz algebra. By the above corollary, if $r$ is a classical Leibniz $r$-matrix in a quadratic Leibniz algebra $(\g,[\cdot,\cdot]_\g,\omega)$, then $r^{\sharp}\circ\omega^{\natural}$ is an $R_+$-matrix introduced in \cite{Rmatrix}.
\end{rmk}

 Our bialgebra theory for Leibniz algebras enjoys many good properties parallelling to that for Lie algebras. This justifies its correctness.

\section{Solutions of the classical Leibniz Yang-Baxter equations}\label{sec:S}

In this section, first we show that a relative Rota-Baxter operator on a Leibniz $(\g,[\cdot,\cdot]_\g)$ with respect to a general representation $(V;\rho^L,\rho^R)$ gives rise to a solution of the classical Leibniz Yang-Baxter equation in a larger Leibniz algebra. Then we introduce the notion of a Leibniz-dendriform algebra, which is the underlying algebraic structure of a relative Rota-Baxter operator on a Leibniz algebra. This type of algebras  play important role in our study of the classical Leibniz Yang-Baxter equation.
There is a natural solution of the classical Leibniz Yang-Baxter equation in the semidirect product Leibniz algebra $A\ltimes_{L_\lhd^*,-L_\lhd^*-R_\rhd^*} A^*$ associated to a Leibniz-dendriform algebra $(A,\triangleright,\triangleleft)$.

\begin{lem}\label{important-dual-rep}
Let $(V;\rho^L,\rho^R)$ be a representation of a Leibniz algebra $(\g,[\cdot,\cdot]_\g)$. Then, the dual representation $(\g^*\oplus V;L^*_{\ltimes},-L^*_{\ltimes}-R^*_{\ltimes})$ of
the regular representation $(\g\oplus V^*;L_\ltimes,R_\ltimes)$ of the semidirect product Leibniz algebra $\g\ltimes_{(\rho^L)^*,-(\rho^L)^*-(\rho^R)^*}V^*$ have the following properties:
$$\begin{array}{llll}
L^*_{\ltimes}(x)\xi=L^*_x\xi\in\g^*,& L^*_{\ltimes}(x)v=\rho^L(x)v\in V, & L^*_{\ltimes}(\chi)\xi=0,& L^*_{\ltimes}(\chi)v\in\g^*,\\
(-L^*_{\ltimes}-R^*_{\ltimes})(x)\xi=(-L^*_x-R^*_x)\xi,&(-L^*_{\ltimes}-R^*_{\ltimes})(x)v=\rho^R(x)v\in V,&&
\\
(-L^*_{\ltimes}-R^*_{\ltimes})(\chi)\xi=0,&(-L^*_{\ltimes}-R^*_{\ltimes})(\chi)v\in\g^*,&&
\end{array}
$$
for all $x\in\g,~v\in V,~\xi\in\g^*,~\chi\in V^*$.
\end{lem}

\begin{thm}\label{o-operator-big-algebra}
A linear map $K:V\lon\g$ is a relative Rota-Baxter operator on a Leibniz algebra $(\g,[\cdot,\cdot]_\g)$ with respect to a representation
$(V;\rho^L,\rho^R)$ if and only if $K+K^*$ is a relative Rota-Baxter operator on the Leibniz algebra $\g\ltimes_{(\rho^L)^*,-(\rho^L)^*-(\rho^R)^*}V^*$ with respect to the dual representation $(\g^*\oplus V;L^*_{\ltimes},-L^*_{\ltimes}-R^*_{\ltimes})$ of
the regular representation $(\g\oplus V^*;L_\ltimes,R_\ltimes)$, that is, the tensor form $\overline{K+K^*}=\Upsilon(K+K^*)$ is a solution of the classical Leibniz Yang-Baxter equation in the Leibniz algebra $\g\ltimes_{(\rho^L)^*,-(\rho^L)^*-(\rho^R)^*}V^*$.
\end{thm}

\begin{proof}
Let $K:V\lon\g$ be a relative Rota-Baxter operator on a Leibniz algebra $(\g,[\cdot,\cdot]_\g)$ with respect to a representation
$(V;\rho^L,\rho^R)$. By Lemma \ref{important-dual-rep}, for all $u,v\in V$, we have
\begin{eqnarray}
\nonumber&&(K+K^*)\Big(L^*_{\ltimes}((K+K^*)u)v-(L^*_{\ltimes}+R^*_{\ltimes})((K+K^*)v)u \Big)-[(K+K^*)u,(K+K^*)v]_{\ltimes}\\
\nonumber&=&(K+K^*)\big(\rho^L(Ku)v+\rho^R(Kv)u\big)-[Ku,Kv]_{\ltimes}\\
\nonumber&=&K\big(\rho^L(Ku)v+\rho^R(Kv)u\big)-[Ku,Kv]_{\g}\\
\label{eq:Olarge}&=&0.
\end{eqnarray}
For all $u,v\in V,\xi\in\g^*$, we have
\begin{eqnarray*}
&&\langle(K+K^*)\Big(L^*_{\ltimes}((K+K^*)u)\xi+(-L^*_{\ltimes}-R^*_{\ltimes})((K+K^*)\xi)u\Big)-[(K+K^*)u,(K+K^*)\xi]_{\ltimes},v\rangle\\
&&=\langle K^*\big(L^*_{Ku}\xi+(-L^*_{\ltimes}-R^*_{\ltimes})(K^*\xi)u\big)-(\rho^L)^*(Ku)K^*\xi,v\rangle\\
&&=\langle L^*_{Ku}\xi+(-L^*_{\ltimes}-R^*_{\ltimes})(K^*\xi)u,Kv\rangle+\langle K^*\xi,\rho^L(Ku)v\rangle\\
&&=-\langle \xi,[Ku,Kv]_\g\rangle+\langle u,[K^*\xi,Kv]_{\ltimes}\rangle+\langle u,[Kv,K^*\xi]_{\ltimes}\rangle+\langle \xi,K(\rho^L(Ku)v)\rangle\\
&&=-\langle \xi,[Ku,Kv]_\g\rangle-\langle u,(\rho^L)^*(Kv)K^*\xi+(\rho^R)^*(Kv)K^*\xi\rangle+\langle u,(\rho^L)^*(Kv)K^*\xi\rangle+\langle \xi,K(\rho^L(Ku)v)\rangle\\
&&=-\langle \xi,[Ku,Kv]_\g\rangle-\langle u,(\rho^R)^*(Kv)K^*\xi\rangle+\langle \xi,K(\rho^L(Ku)v)\rangle\\
&&=-\langle \xi,[Ku,Kv]_\g\rangle+\langle K(\rho^R(Kv)u),\xi\rangle+\langle \xi,K(\rho^L(Ku)v)\rangle\\
&&=0.
\end{eqnarray*}
Similarly, we can show that for all $X,Y\in\g^*\oplus V,$ we have
$$
(K+K^*)\Big(L^*_{\ltimes}((K+K^*)X)Y+(-L^*_{\ltimes}-R^*_{\ltimes})((K+K^*)Y)X\Big)-[(K+K^*)X,(K+K^*)Y]_{\ltimes}=0,
$$
which implies that $K+K^*$ is a relative Rota-Baxter operator on the Leibniz algebra $\g\ltimes_{(\rho^L)^*,-(\rho^L)^*-(\rho^R)^*}V^*$ with respect to the   representation  $(\g^*\oplus V;L^*_{\ltimes},-L^*_{\ltimes}-R^*_{\ltimes})$.

\emptycomment{for $u,v\in V,\xi\in\g^*$, we have
\begin{eqnarray*}
&&\langle(T+T^*)\big(L^*_{\ltimes}((T+T^*)\xi)u-L^*_{\ltimes}((T+T^*)u)\xi-R^*_{\ltimes}((T+T^*)u)\xi\big)-[(T+T^*)\xi,(T+T^*)u]_{\ltimes},v\rangle\\
&=&\langle T^*\big(L^*_{\ltimes}(T^*\xi)u-(L^*_{Tu}+R^*_{Tu})\xi\big)+(\rho^L)^*(Tu)T^*\xi+(\rho^R)^*(Tu)T^*\xi,v\rangle\\
&=&\langle L^*_{\ltimes}(T^*\xi)u-(L^*_{Tu}+R^*_{Tu})\xi,Tv\rangle-\langle T^*\xi,\rho^L(Tu)v\rangle-\langle T^*\xi,\rho^R(Tu)v\rangle\\
&=&-\langle u,[T^*\xi,Tv]_{\ltimes}\rangle+\langle\xi,[Tu,Tv]_\g\rangle+\langle\xi,[Tv,Tu]_\g\rangle-\langle \xi,T(\rho^L(Tu)v)\rangle-\langle \xi,T(\rho^R(Tu)v)\rangle\\
&=&\langle u,(\rho^L)^*(Tv)T^*\xi+(\rho^R)^*(Tv)T^*\xi\rangle+\langle\xi,[Tu,Tv]_\g\rangle+\langle\xi,[Tv,Tu]_\g\rangle-\langle \xi,T(\rho^L(Tu)v)\rangle-\langle \xi,T(\rho^R(Tu)v)\rangle\\
&=&-\langle T(\rho^L(Tv)u),\xi\rangle-\langle T(\rho^R(Tv)u),\xi\rangle+\langle\xi,[Tu,Tv]_\g\rangle+\langle\xi,[Tv,Tu]_\g\rangle-\langle \xi,T(\rho^L(Tu)v)\rangle-\langle \xi,T(\rho^R(Tu)v)\rangle\\
&=&0.
\end{eqnarray*}
For $\xi,\eta\in\g^*$, we have
\begin{eqnarray*}
&&(T+T^*)\big(L^*_{\ltimes}((T+T^*)\xi)\eta-L^*_{\ltimes}((T+T^*)\eta)\xi-R^*_{\ltimes}((T+T^*)\eta)\xi\big)-[(T+T^*)\xi,(T+T^*)\eta]_{\ltimes}\\
&=&(T+T^*)\big(L^*_{\ltimes}(T^*\xi)\eta-L^*_{\ltimes}(T^*\eta)\xi-R^*_{\ltimes}(T^*\eta)\xi\big)-[T^*\xi,T^*\eta]_{\ltimes}\\
&=&0.
\end{eqnarray*}
}

Conversely, let $K+K^*$ be a relative Rota-Baxter operator on the Leibniz algebra $\g\ltimes_{(\rho^L)^*,-(\rho^L)^*-(\rho^R)^*}V^*$ with respect to the   representation $(\g^*\oplus V;L^*_{\ltimes},-L^*_{\ltimes}-R^*_{\ltimes})$. By \eqref{eq:Olarge},
we deduce that $K:V\lon\g$ is a relative Rota-Baxter operator on   $(\g,[\cdot,\cdot]_\g)$ with respect to the representation
$(V;\rho^L,\rho^R)$.
\end{proof}

In the sequel, we introduce the notion of a Leibniz-dendriform algebra as the underlying algebraic structure of a relative Rota-Baxter operator on a Leibniz algebra.

\begin{defi}\label{Leibniz-dendriform}
A {\bf Leibniz-dendriform algebra} is a vector space $A$ equipped with two binary operations  $\rhd$ and $\lhd:A\otimes A\lon A$
 such that for all $x,y,z\in A,$ we have
\begin{eqnarray}
\label{p1}(x\lhd y)\lhd z&=&x\lhd(y\lhd z)-y\lhd(x\lhd z)-(x\rhd y)\lhd z,\\
\label{p2}x\lhd(y\rhd z) &=&(x\lhd y)\rhd z+y\rhd(x\lhd z)+y\rhd(x\rhd z),\\
\label{p3}x\rhd(y\rhd z) &=&(x\rhd y)\rhd z+y\lhd(x\rhd z)-x\rhd(y\lhd z).
\end{eqnarray}
\end{defi}

\begin{pro}\label{Leibniz-up-Leibniz-dendriform}
Let $(A,\rhd,\lhd)$ be a Leibniz-dendriform algebra. Then the binary operation $[\cdot,\cdot]_{\rhd,\lhd}:A\otimes A\lon A$ given by
\begin{eqnarray}
[x,y]_{\rhd,\lhd}=x\lhd y+x\rhd y,\,\,\,\,\forall x,y\in A,
\end{eqnarray}
defines a Leibniz algebra, which is called the {\bf sub-adjacent Leibniz algebra} of $(A,\rhd,\lhd)$ and $(A,\rhd,\lhd)$ is called a {\bf
compatible Leibniz-dendriform algebra} structure on  $(A,[\cdot,\cdot]_{\rhd,\lhd})$.
\end{pro}
\begin{proof}
For all $x,y,z\in A$, we have
\begin{eqnarray*}
[x,[y,z]_{\rhd,\lhd}]_{\rhd,\lhd}=[x,y\lhd z+y\rhd z]_{\rhd,\lhd}=x\lhd (y\lhd z)+x\rhd (y\lhd z)+x\lhd (y\rhd z)+x\rhd (y\rhd z).
\end{eqnarray*}
On the other hand, we have
\begin{eqnarray*}
[[x,y]_{\rhd,\lhd},z]_{\rhd,\lhd}+[y,[x,z]_{\rhd,\lhd}]_{\rhd,\lhd}&=&[x\lhd y+x\rhd y,z]_{\rhd,\lhd}+[y,x\lhd z+x\rhd z]_{\rhd,\lhd}\\
                   &=&(x\lhd y)\lhd z+(x\lhd y)\rhd z+(x\rhd y)\lhd z+(x\rhd y)\rhd z\\
                   &&+y\lhd (x\lhd z)+y\rhd (x\lhd z)+y\lhd (x\rhd z)+y\rhd (x\rhd z).
\end{eqnarray*}
Thus, $(A,[\cdot,\cdot]_{\rhd,\lhd})$ is a Leibniz algebra.
\end{proof}

\begin{ex}\label{example-7}
{\rm
Let $V$ be a vector space. On the vector space $\gl(V)\oplus V$, define two binary operations $\rhd$ and $\lhd:(\gl(V)\oplus V)\otimes (\gl(V)\oplus V)\lon \gl(V)\oplus V$ by
\begin{eqnarray*}
(A+u)\lhd(B+v)=AB+Av,\quad (A+u)\rhd(B+v)=-BA,\quad \forall A,B\in\gl(V),~ u,v\in V.
\end{eqnarray*}
Then $(\gl(V)\oplus V,\rhd,\lhd)$ is a Leibniz-dendriform algebra. Its sub-adjacent Leibniz algebra is exactly the one underlying an omni-Lie algebra introduced by Weinstein in \cite{Alan}.
}
\end{ex}

Let $(A,\rhd,\lhd)$ be a Leibniz-dendriform algebra. Define two linear maps $L_\lhd:A\lon\gl(A)$ and $R_\rhd:A\lon\gl(A)$ by
\begin{eqnarray}
L_\lhd(x)y=x\lhd y,\,\,\,\,R_\rhd(x)y=y\rhd x,\,\,\,\,\forall x,y\in A.
\end{eqnarray}

\begin{pro}\label{adjacent-rep}
Let $(A,\rhd,\lhd)$ be a Leibniz-dendriform algebra. Then $(A;L_\lhd,R_\rhd)$ is a representation of the sub-adjacent Leibniz algebra $(A,[\cdot,\cdot]_{\rhd,\lhd})$. Moreover, the identity map ${\Id}:A\lon A$ is a relative Rota-Baxter operator on the Leibniz algebra $(A,[\cdot,\cdot]_{\rhd,\lhd})$ with respect to the representation $(A;L_\lhd,R_\rhd)$.
\end{pro}

\begin{proof}
By \eqref{p1}, for all $x,y,z\in A$, we have
$$
(L_\lhd([x,y]_{\rhd,\lhd})-[L_\lhd(x),L_\lhd(y)])z=[x,y]_{\rhd,\lhd}\lhd z-x\lhd(y\lhd z)+y\lhd(x\lhd z)=0.
$$
Thus, we have $L_\lhd([x,y]_{\rhd,\lhd})=[L_\lhd(x),L_\lhd(y)]$. By \eqref{p2},   we have
$$
(R_\rhd([x,y]_{\rhd,\lhd})-[L_\lhd(x),R_\rhd(y)])z=z\rhd[x,y] -x\lhd(z\rhd y)+(x\lhd z)\rhd y=0,
$$
which implies that $R_\rhd([x,y]_{\rhd,\lhd})=[L_\lhd(x),R_\rhd(y)]$. By \eqref{p2} and \eqref{p3},   we have
\begin{eqnarray*}
(R_\rhd(y)L_\lhd(x)+R_\rhd(y)R_\rhd(x))z&=&(x\lhd z)\rhd y+(z\rhd x)\rhd y\\
                                        &=&x\lhd(z\rhd y)-z\rhd(x\lhd y)-z\rhd(x\rhd y)\\
                                        &&+z\rhd(x\rhd y)-x\lhd(z\rhd y)+z\rhd(x\lhd y)\\
                                        &=&0.
\end{eqnarray*}
Thus, we have $R_\rhd(y)L_\lhd(x)=-R_\rhd(y)R_\rhd(x)$. Therefore, $(A;L_\lhd,R_\rhd)$ is a representation of the sub-adjacent Leibniz algebra $(A,[\cdot,\cdot]_{\rhd,\lhd})$. Moreover, we have
\begin{eqnarray*}
{\Id}(L_\lhd({\Id} (x))y+R_\rhd({\Id} (y))x)=x\lhd y+x\rhd y=[{\Id} (x),{\Id} (y)]_{\rhd,\lhd}.
\end{eqnarray*}
Thus, we obtain that ${\Id}:A\lon A$ is a relative Rota-Baxter operator on the Leibniz algebra $(A,[\cdot,\cdot]_{\rhd,\lhd})$ with respect to the representation $(A;L_\lhd,R_\rhd)$. The proof is finished.
\end{proof}

\begin{pro}\label{Leibniz-dendriform-under-o-operator}
  Let $K:V\lon\g$ be a relative Rota-Baxter operator on  a Leibniz algebra $(\g,[\cdot,\cdot]_\g)$ with respect to a representation $(V;\rho^L,\rho^R)$. Then there is a Leibniz-dendriform algebra structure on $V$ given by
\begin{eqnarray}
u \rhd v:=\rho^R(Kv)u,\,\,\,\,u \lhd v:=\rho^L(Ku)v,\,\,\,\,\forall u,v\in V.
\end{eqnarray}
\end{pro}

\begin{proof}
By \eqref{rep-1} and \eqref{O-operator}, we have
\begin{eqnarray*}
&&u\lhd(v\lhd w)-v\lhd(u\lhd w)-(u\rhd v)\lhd w-(u \lhd v)\lhd w\\
&=&u\lhd(\rho^L(Kv)w)-v\lhd(\rho^L(Ku)w)-(\rho^R(Kv)u)\lhd w-(\rho^L(Ku)v)\lhd w\\
&=&\rho^L(Ku)\rho^L(Kv)w-\rho^L(Kv)\rho^L(Ku)w-\rho^L(K(\rho^R(Kv)u))w-\rho^L(K(\rho^L(Ku)v))w\\
&=&\rho^L([Ku,Kv]_\g)w-\rho^L(K(\rho^R(Kv)u))w-\rho^L(K(\rho^L(Ku)v))w\\
&=&0,
\end{eqnarray*}
which implies that \eqref{p1} in Definition \ref{Leibniz-dendriform} holds.

Similarly, we can show that \eqref{p2} and \eqref{p3} also hold.
Thus,   $(V,\rhd,\lhd)$ is a Leibniz-dendriform algebra.
\end{proof}

\emptycomment{
By \eqref{rep-2} and \eqref{O-operator}, we have
\begin{eqnarray*}
&&(u\lhd v)\rhd w+v\rhd(u\lhd w)+v\rhd(u\rhd w)-u \lhd (v\rhd w)\\
&=&(\rho^L(Tu)v)\rhd w+v\rhd(\rho^L(Tu)w)+v\rhd(\rho^R(Tw)u)-u \lhd (\rho^R(Tw)v)\\
&=&\rho^R(Tw)\rho^L(Tu)v+\rho^R(T(\rho^L(Tu)w))v+\rho^R(T(\rho^R(Tw)u))v-\rho^L(Tu)\rho^R(Tw)v\\
&=&\rho^R(T(\rho^L(Tu)w))v+\rho^R(T(\rho^R(Tw)u))v-\rho^R([Tu,Tw]_\g)v\\
&=&0.
\end{eqnarray*}

By \eqref{rep-2}, \eqref{rep-3} and \eqref{O-operator}, we have
\begin{eqnarray*}
&&(u\rhd v)\rhd w+v\lhd(u\rhd w)-u\rhd(v\lhd w)-u\rhd(v\rhd w)\\
&=&(\rho^R(Tv)u)\rhd w+v\lhd(\rho^R(Tw)u)-u\rhd(\rho^L(Tv)w)-u\rhd(\rho^R(Tw)v)\\
&=&\rho^R(Tw)\rho^R(Tv)u+\rho^L(Tv)\rho^R(Tw)u-\rho^R(T(\rho^L(Tv)w))u-\rho^R(T(\rho^R(Tw)v))u\\
&=&-\rho^R(Tw)\rho^L(Tv)u+\rho^L(Tv)\rho^R(Tw)u-\rho^R(T(\rho^L(Tv)w))u-\rho^R(T(\rho^R(Tw)v))u\\
&=&\rho^R([Tv,Tw]_\g)u-\rho^R(T(\rho^L(Tv)w))u-\rho^R(T(\rho^R(Tw)v))u\\
&=&0.
\end{eqnarray*}
}

We give a sufficient and necessary condition for the existing of a compatible Leibniz-dendriform algebra structure on a Leibniz algebra.

\begin{pro}
Let $(\g,[\cdot,\cdot]_\g)$ be a Leibniz algebra. Then there is a compatible Leibniz-dendriform algebra on $\g$ if and only if there exists an invertible relative Rota-Baxter operator $K:V\lon\g$ on $\g$ with respect to a representation $(V;\rho^L,\rho^R)$. Furthermore, the compatible Leibniz-dendriform algebra structure on $\g$ is given by
\begin{eqnarray}
x \rhd y:=K(\rho^R(y)K^{-1}x),\,\,\,\,x \lhd y:=K(\rho^L(x)K^{-1}y),\,\,\,\,\forall x,y\in \g.
\end{eqnarray}
\end{pro}

\begin{proof}
Let $K:V\lon\g$ be an invertible  relative Rota-Baxter operator on $\g$ with respect to a representation $(V;\rho^L,\rho^R)$. By Proposition \ref{Leibniz-dendriform-under-o-operator}, there is a Leibniz-dendriform algebra on $V$ given by
$$
u \rhd v:=\rho^R(Kv)u,\,\,\,\,u \lhd v:=\rho^L(Ku)v,\,\,\,\,\forall u,v\in V.
$$
 Since $K$ is an invertible  relative Rota-Baxter operator, we obtain that
$$
x \rhd y:=K(K^{-1}x\rhd K^{-1}x)=K(\rho^R(y)K^{-1}x),\,\,\,\,x \lhd y:=K(K^{-1}x\lhd K^{-1}y)=K(\rho^L(x)K^{-1}y),
$$
is a Leibniz-dendriform algebra on $\g$. By \eqref{O-operator}, we have
\begin{eqnarray*}
x \rhd y+x \lhd y&=&K(\rho^R(y)K^{-1}x)+K(\rho^L(x)K^{-1}y)\\
                 &=&K(\rho^R(K(K^{-1}y))K^{-1}x)+K(\rho^L(K(K^{-1}x))K^{-1}y)\\
                 &=&[x,y]_\g.
\end{eqnarray*}

On the other hand, let $(\g,\rhd,\lhd)$ be a compatible Leibniz-dendriform algebra of the Leibniz algebra $(\g,[\cdot,\cdot]_\g)$. By Proposition \ref{adjacent-rep}, $(\g;L_\lhd,R_\rhd)$ is a representation of the Leibniz algebra $(\g,[\cdot,\cdot]_\g)$. Moreover, ${\Id}:\g\lon\g$ is a  relative Rota-Baxter operator on the Leibniz algebra $(\g,[\cdot,\cdot]_\g)$ with respect to the representation $(\g;L_\lhd,R_\rhd)$. The proof is finished.
\end{proof}

\emptycomment{
The algebraic structures underlying such bilinear forms $\huaB\in\Sym^2(\g^*)$ are Leibniz-dendriform algebras.

\begin{pro}
Let $\huaB\in\Sym^2(\g^*)$ be a nondegenerate bilinear form satisfying  \eqref{Hessian-structure}. Then there is a compatible Leibniz-dendriform algebra structure on
$\g$ given by:
\begin{eqnarray}
\huaB(x\rhd y,z)=\huaB(x,[y,z]_\g+[z,y]_\g),\quad\huaB(x\lhd y,z)=-\huaB(y,[x,z]_\g),\quad\forall x,y,z\in\g.
\end{eqnarray}
\end{pro}

\begin{proof}
Let $\huaB\in\Sym^2(\g^*)$ be a nondegenerate bilinear form satisfying  \eqref{Hessian-structure}. Then $K=(\huaB^{\natural})^{-1}:\g^*\lon\g$ is an invertible relative Rota-Baxter operator on $\g$ with respect to the representation $(\g^*;L^*,-L^*-R^*)$ and $K^*=K$. Thus, we have a compatible Leibniz-dendriform algebra structure on
$\g$ given by:
\begin{eqnarray}
x \rhd y:=-(\huaB^{\natural})^{-1}(L^*_y\huaB^{\natural}x)-(\huaB^{\natural})^{-1}(R^*_y\huaB^{\natural}x),\,\,\,\,x \lhd y:=(\huaB^{\natural})^{-1}(L^*_x\huaB^{\natural}y),\,\,\,\,\forall x,y\in \g.
\end{eqnarray}
Therefore we have
\begin{eqnarray*}
\huaB(x\rhd y,z)&=&\langle\huaB^{\natural}(x\rhd y),z\rangle=-\langle(L^*_y\huaB^{\natural}x),z\rangle-\langle(R^*_y\huaB^{\natural}x),z\rangle=\huaB(x,[y,z]_\g)+\huaB(x,[z,y]_\g),\\
\huaB(x\lhd y,z)&=&\langle\huaB^{\natural}(x\lhd y),z\rangle=\langle L^*_x\huaB^{\natural}y,z\rangle=-\langle \huaB^{\natural}y,[x,z]_\g\rangle=-\huaB(y,[x,z]_\g).
\end{eqnarray*}
The proof is finished.
\end{proof}
}

\begin{thm}
Let $(A,\rhd,\lhd)$ be a Leibniz-dendriform algebra. Then
\begin{eqnarray}
r:=\sum_{i=1}^{n}(e_i^*\otimes e_i+e_i\otimes e_i^*)
\end{eqnarray}
is a symmetric solution of the classical Leibniz Yang-Baxter equation in the Leibniz algebra $A\ltimes_{L_\lhd^*,-L_\lhd^*-R_\rhd^*}A^*$,  where $\{e_1,\cdots,e_n\}$ is a basis of
$A$ and $\{e_1^*,\cdots,e_n^*\}$ is its dual basis. Moreover, $r$ is nondegenerate and the induced bilinear form $\huaB$ on $A\ltimes_{L_\lhd^*,-L_\lhd^*-R_\rhd^*}A^*$ is
given by:
\begin{eqnarray}\label{eq:specialB}
\huaB(x+\xi,y+\eta)=\langle \xi,y\rangle+\langle \eta,x\rangle.
\end{eqnarray}
\end{thm}

\begin{proof}
Since $(A,\rhd,\lhd)$ is a Leibniz-dendriform algebra, the identity map ${\Id}:A\lon A$ is a relative Rota-Baxter operator on the sub-adjacent Leibniz algebra $(A,[\cdot,\cdot]_{\rhd,\lhd})$ with respect to the representation $(A;L_\lhd,R_\rhd)$. By Theorem \ref{o-operator-big-algebra}, $r=\sum_{i=1}^{n}(e_i^*\otimes e_i+e_i\otimes e_i^*)$ is a symmetric solution of the classical Leibniz Yang-Baxter equation in $A\ltimes_{L_\lhd^*,-L_\lhd^*-R_\rhd^*}A^*$. It is obvious that the corresponding bilinear form $\huaB\in \Sym^2(A\oplus A^*)$ is given by  \eqref{eq:specialB}.
The proof is finished.
\end{proof}

The above  results can be viewed as the Leibniz analogue of the results given in \cite{Bai-1}.

\vspace{2mm}
\noindent
{\bf Acknowledgements. } This research is supported by NSFC (11922110). We give warmest thanks to Xiaomeng Xu for helpful comments.

 \end{document}